\documentclass[11pt]{article}

\usepackage[top=1in,bottom=1in,left=1in,right=1in]{geometry}
\usepackage{natbib}
\usepackage[unicode=true,bookmarks=true,bookmarksnumbered=false,bookmarksopen=false,breaklinks=false,pdfborder={0 0 1},backref=false,colorlinks=true]{hyperref}
\hypersetup{citecolor=blue}
\usepackage{url}

\usepackage{amsmath}
\usepackage{mathrsfs}
\usepackage{amssymb}
\usepackage{amsthm}
\usepackage[normalem]{ulem}

\newtheorem{theorem}{Theorem}

\newtheorem{proposition}{Proposition}

\usepackage{float}
\usepackage{rotfloat}

\floatstyle{ruled}
\newfloat{algorithm}{tpb}{loa}
\providecommand{\algorithmname}{Algorithm}
\floatname{algorithm}{\protect\algorithmname}

\usepackage{algpseudocode,algorithm}
\algnewcommand\algorithmicinput{\textbf{Input}:}
\algnewcommand\algorithmicoutput{\textbf{Output}:}
\algnewcommand\INPUT{\item[\algorithmicinput]}
\algnewcommand\OUTPUT{\item[\algorithmicoutput]}
\usepackage{graphicx}
\usepackage[caption=false,labelformat=simple]{subfig}

\usepackage{array}
\newcolumntype{L}[1]{>{\raggedright\let\newline\\\arraybackslash\hspace{0pt}}m{#1}}
\newcolumntype{C}[1]{>{\centering\let\newline\\\arraybackslash\hspace{0pt}}m{#1}}
\newcolumntype{R}[1]{>{\raggedleft\let\newline\\\arraybackslash\hspace{0pt}}m{#1}}

\usepackage{rotating}
\usepackage{multirow}

\usepackage{tikz}
\usetikzlibrary{shapes,arrows,backgrounds,calc,positioning,fit,petri,plotmarks}
\usetikzlibrary{arrows}

\usepackage{enumerate}
\usepackage{paralist}

\newcommand*{\affaddr}[1]{#1} % No op here. Customize it for different styles.
\newcommand*{\affmark}[1][*]{\textsuperscript{#1}}

\global\long\def\bX{\mathbf{X}}
\global\long\def\bx{\mathbf{x}}
\global\long\def\bY{\mathbf{Y}}
\global\long\def\by{\mathbf{y}}

\global\long\def\bA{\mathbf{A}}

\global\long\def\bH{\mathbf{H}}

\global\long\def\bw{\mathbf{w}}

\global\long\def\bQ{\mathbf{Q}}

\global\long\def\bS{\mathbf{S}}

\global\long\def\bbeta{\boldsymbol{\beta}}

\global\long\def\bSigma{\boldsymbol{\Sigma}}
\global\long\def\bgamma{\boldsymbol{\gamma}}

\global\long\def\btheta{\boldsymbol{\theta}}

\global\long\def\bLambda{\boldsymbol{\Lambda}}
\global\long\def\bDelta{\boldsymbol{\Delta}}

\global\long\def\bPi{\boldsymbol{\Pi}}
\global\long\def\bGamma{\boldsymbol{\Gamma}}
\global\long\def\bTheta{\boldsymbol{\Theta}}
\global\long\def\bUpsilon{\boldsymbol{\Upsilon}}
\global\long\def\bOmega{\boldsymbol{\Omega}}
\global\long\def\bpi{\boldsymbol{\pi}}
\global\long\def\bupsilon{\boldsymbol{\upsilon}}
\global\long\def\bzeta{\boldsymbol{\zeta}}
\global\long\def\bkappa{\boldsymbol{\kappa}}

\usepackage{xr}
\makeatletter
\newcommand*{\addFileDependency}[1]{% argument=file name and extension
  \typeout{(#1)}
  \@addtofilelist{#1}
  \IfFileExists{#1}{}{\typeout{No file #1.}}
}
\makeatother
 
\newcommand*{\myexternaldocument}[1]{%
    \externaldocument{#1}%
    \addFileDependency{#1.tex}%
    \addFileDependency{#1.aux}%
}

\myexternaldocument{COCReg_sup_221011}
\title{Covariance-on-Covariance Regression}

\author{%
    Yi Zhao\affmark[1] and Yize Zhao\affmark[2] \\
    \affaddr{\affmark[1]Department of Biostatistics and Health Data Science, Indiana University School of Medicine} \\
    \affaddr{\affmark[2]Department of Biostatistics, Yale University School of Medicine} \\
}

\date{}

\providecommand{\keywords}[1]
{
  {\small   
  \textbf{Keywords:} #1 }
}

\begin{document}
%%%%%%%%%%%%%%%%%%%%%%%%%%%%%%%%%%%%%%%%%%%%%%%%%%%%%%%%%%

\maketitle

\thispagestyle{empty}

\begin{abstract}
    A Covariance-on-Covariance regression model is introduced in this manuscript. It is assumed that there exists (at least) a pair of linear projections on outcome covariance matrices and predictor covariance matrices such that a log-linear model links the variances in the projection spaces, as well as additional covariates of interest. An ordinary least square type of estimator is proposed to simultaneously identify the projections and estimate model coefficients. Under regularity conditions, the proposed estimator is asymptotically consistent. The superior performance of the proposed approach over existing methods are demonstrated via simulation studies. Applying to data collected in the Human Connectome Project Aging study, the proposed approach identifies three pairs of brain networks, where functional connectivity within the resting-state network predicts functional connectivity within the corresponding task-state network. The three networks correspond to a global signal network, a task-related network, and a task-unrelated network. The findings are consistent with existing knowledge about brain function.
\end{abstract}

%%%%%%%%%%%%%%%%%%%%%%%%%%%%%%%%%%%%%%%%%%%%%%%%%%%%%%%%%%
%========================================================%

%========================================================%

%%%%%%%%%%%%%%%%%%%%%%%%%%%%%%%%%%%%%%%%%%%%%%%%%%%%%%%%%%
%========================================================%
\clearpage
\setcounter{page}{1}

%%%%%%%%%%%%%%%%%%%%%%%%%%%%%%%%%%%%%%%%%%%%%%%%%%%%%%%%%%
% Introduction
%%%%%%%%%%%%%%%%%%%%%%%%%%%%%%%%%%%%%%%%%%%%%%%%%%%%%%%%%%
\section{Introduction}
\label{sec:intro}

In this manuscript, a Covariance-on-Covariance regression problem is studied. This is motivated by functional magnetic resonance imaging (fMRI) experiments. Typically, there are two types of fMRI experiments, resting-state fMRI (rs-fMRI) and task-based fMRI (tb-fMRI). In a rs-fMRI experiment, the participants are asked to lie in the scanner at rest with eyes open. The study interest is to characterize the coactivation pattern between brain regions, the so-called brain functional connectivity,  captured by the correlation, or the covariance after properly standardizing the data, between fMRI time courses. In a tb-fMRI experiment, the participants are instructed to complete a certain task during the imaging scan. Examples include finger tapping task, go/no-go task, $n$-back working memory task, and so on. Such type of fMRI experiments attempts to identify activated brain regions in response to the task stimuli and depict task-state functional networks. Building upon the assumption that individual variation in brain responses to a task is intrinsic in the brain and predictable, existing research demonstrates that one can use task-free measurements to predict task-related activations~\citep{tavor2016task,lacosse2021jumping,cohen2020regression,ngo2022predicting}. An on-going research topic is to find the model that yields the optimal performance in predicting the task activation map using the rs-features. 
In tb-fMRI experiments, studying the task-state functional connectivity is also important. When performing cognitive tasks, small alterations from the intrinsic network organization occur and strongly contribute to task performance~\citep{cole2021functional}, where the intrinsic network organization can be captured by the functional connectivity map using rs-fMRI. Considering the fact that cognitive task activations emerge through network interactions, predicting task-state functional connectivity from resting-state functional connectivity is feasible and promising.
Putting it into a statistical framework, it is a regression problem with a covariance matrix (tb-connectivity) as the outcome and a covariance matrix (rs-connectivity) as one of the predictors.

Popular approaches of using the rs-connectivity features to predict tb-connectivity include the general linear model and statistical/machine learning techniques enriching the collection of literature on ``connectome fingerprinting''~\citep{mennes2010inter,tavor2016task}. After choosing a brain parcellation, rs-connectivity and tb-connectivity are calculated for each pair of brain regions. One straightforward approach is to fit models to each pair of tb-connectivity and loop over all possible pairs, where in each model, all pairs of rs-connectivity can be entered as the predictors. This approach, however, ignores the structure of the connectivity matrices (both tb-connectivity and rs-connectivity), such as the positive definiteness and topological architecture, and suffers from a deficient power due to multiplicity. Considering covariance matrices as the outcome, to preserve the structural property and directly characterize covariance matrices with covariates of interest, a type of regression model, called covariance regression, was introduced. Examples include modeling the covariance matrix as a quadratic function of the covariates~\citep{hoff2012covariance,seiler2017multivariate} or a linear combination of similarity matrices of the covariates~\citep{zou2017covariance}, nonparametric covariance regression utilizing low-rank approximation~\citep{fox2015bayesian}, common diagonalization based on eigendecomposition~\citep{flury1984common,boik2002spectral,hoff2009hierarchical,franks2019shared}  and Cholesky decomposition~\citep{pourahmadi2007simultaneous}, and a recent approach of covariate assisted principal regression for covariance matrix outcomes~\citep{zhao2021covariate}.
In \citet{zhao2021covariate}, it is assumed that there exists a common diagonalization on the covariance matrices and the corresponding diagonal elements satisfy a log-linear model on the covariates of interest. The advantage is to preserve the positive definiteness of the covariance matrices and offer high flexibility in parsimonious modeling. 
In this study, building upon this idea of common diagonalization, an approach to perform \textbf{C}ovariance-\textbf{o}n-\textbf{C}ovariance \textbf{Reg}ression (CoCReg) is introduced. Common linear projections are assumed for both outcome covariance matrices and predictor covariance matrices and a log-linear regression model is assumed for the projected data and the rest scalar covariates of interest. The objective is to identify the projections and simultaneously estimate the model coefficients. With proper thresholding or sparsifying on the loading profiles, the model offers a network-level interpretation, that is the resting-state functional connectivity within the network can predict the connectivity between regions within the corresponding task-based network. For tb-fMRI signals, normality cannot be assumed. Thus, likelihood-based estimators cannot be employed. Instead, an ordinary least squares type of estimator is proposed and asymptotic consistency can be achieved. 

Another group of related work is the image-on-image regression. Here, an image means a vector of scalar outcomes with spatial information. Compared with image-on-scalar and scalar-on-image regression, image-on-image has been less explored but received growing attention along the need in modern neuroscientific research. Recently, \citet{guo2022spatial} developed a spatial Bayesian latent factor model to predict individual task-evoked images using the corresponding task-independent images. Data dimension was significantly reduced using proper basis functions and choosing a small number of latent factors, and at the same time taking spatial dependence into consideration. Some earlier attempts include \citet{sweeney2013automatic}, which considered voxel-wise regression models for imaging prediction, however, ignored the spatial correlations. Later, \citet{hazra2019spatio} also proposed to perform voxel-wise models but to include effects from the neighboring voxels. 
In this paper, the proposed framework can be viewed as an image network-on-image network regression, which is sharply distinguished from all the existing works.

The rest of this manuscript is organized as the following. Section~\ref{sec:model} introduces the proposed Covariance-on-Covariance regression model based on linear projections. A least-square type of estimator is introduced to simultaneously estimate the linear projections and model coefficients. The asymptotic consistency of the proposed estimator is discussed under regularity conditions. The performance of the proposed approach is evaluated via simulation studies in Section~\ref{sec:sim} and the task/resting-state fMRI data collected in the Lifespan Human Connectome Project Aging (HCP-A) Study in Section~\ref{sec:fmri}. Section~\ref{sec:discussion} summarizes the manuscript with discussions.
%%%%%%%%%%%%%%%%%%%%%%%%%%%%%%%%%%%%%%%%%%%%%%%%%%%%%%%%%%

%%%%%%%%%%%%%%%%%%%%%%%%%%%%%%%%%%%%%%%%%%%%%%%%%%%%%%%%%%
% Model
%%%%%%%%%%%%%%%%%%%%%%%%%%%%%%%%%%%%%%%%%%%%%%%%%%%%%%%%%%
\section{Model and Method}
\label{sec:model}

%------------------------------------------
Assume data are collected from $n$ subjects. Let $\bx_{is}\in\mathbb{R}^{p}$ denote the $p$-dimensional predictor of the $s$th observation from subject $i$, where $s=1,\dots,u_{i}$ and $u_{i}$ is the number of observations. Let $\by_{it}\in\mathbb{R}^{q}$ denote the $q$-dimensional outcome of the $t$th observation from subject $i$, where $t=1,\dots,v_{i}$ and $v_{i}$ is the number of observations. Denote $\bDelta_{i}\in\mathbb{R}^{p\times p}$ as the covariance matrix of $\bx_{is}$ ($s=1,\dots,u_{i}$) and $\bSigma_{i}\in\mathbb{R}^{q\times q}$ as the covariance matrix of $\by_{it}$ ($t=1,\dots,v_{i}$), for $i=1,\dots,n$. The objective is to generalize the concept of regression to study the association between the two sets of covariance matrices. 
In the data application, $\bx_{is}$'s are the rs-fMRI signals and $\by_{it}$'s are the tb-fMRI signals. Assuming the signals are centralized to mean zero and standardized to unit variance, the covariance matrices represent the resting-state functional connectivity ($\bDelta_{i}$'s) and brain connectivity in task ($\bSigma_{i}$'s), respectively. The study interest is to investigate if the resting-state functional connectivity can predict brain connectivity under an on-going in-scanner task.
Let $\bw_{i}\in\mathbb{R}^{r}$ denote the $r$-dimensional other covariates of interest (with the first element of one for the intercept term). It is assumed that there exists a linear projection on the $\bx_{is}$'s, denoted as $\btheta\in\mathbb{R}^{p}$, and a liner projection on the $\by_{it}$'s, denoted as $\bgamma\in\mathbb{R}^{q}$, such that the following regression model holds.
\begin{equation}\label{eq:model}
    \log\left(\bgamma^\top\bSigma_{i}\bgamma\right)=\alpha\log\left(\btheta^\top\bDelta_{i}\btheta\right)+\bw_{i}^\top\bbeta,
\end{equation}
where $\alpha\in\mathbb{R}$ and $\bbeta\in\mathbb{R}^{r}$ are model coefficients. Denote $\zeta_{it}=\bgamma^\top\by_{it}$. $\bgamma$ projects $\by_{it}$ into $\mathbb{R}$ and $\mathrm{Var}(\zeta_{it})=\bgamma^\top\bSigma_{i}\bgamma$. Analogously, denote $\kappa_{is}=\btheta^\top\bx_{is}$ and $\mathrm{Var}(\kappa_{is})=\btheta^\top\bDelta_{i}\btheta$. Model~\eqref{eq:model} is thus a type of generalized regression model with a logarithmic link on variance components in the projected spaces and to study the association between the two variances. When $\bw_{i}$ contains the intercept term only ($r=1$), the model can be written as
\begin{equation}
    \log\left(\bgamma^\top\bSigma_{i}\bgamma\right)=\beta_{0}+\alpha\log\left(\btheta^\top\bDelta_{i}\btheta\right).
\end{equation}
This can be viewed as a generalization of the canonical correlation analysis (CCA), but to characterize the association between the second-order moments, that is the covariance matrices of $\by$ and $\bx$.
When $\btheta$ is a prespecified projection vector, for example, a subgroup of brain regions with equal weight, plugging in an estimate of $\bDelta_{i}$ (denoted as $\hat{\bDelta}_{i}$) and treating $\log(\btheta^\top\hat{\bDelta}_{i}\btheta)$ as a covariate, Model~\eqref{eq:model} reduces to the covariate assisted principal regression model proposed in \citet{zhao2021covariate}. However, since distributional assumptions on $\by_{it}$ are not imposed, instead of a likelihood-based estimator introduced by \citet{zhao2021covariate}, a distribution-free estimator is introduced in the next section.
When $\bgamma$ is prespecified, Model~\eqref{eq:model} should be distinguished from the principal component regression as the study interest is on the association to the variation of the components, that is, $\mathrm{Var}(\btheta^\top\bx_{is})$, rather than the principal components themselves.
%------------------------------------------

%------------------------------------------
\subsection{Estimation}
\label{sub:estimation}

In this section, an ordinary least square (OLS) type of estimator is introduced to relax the distributional assumptions on the data. For resting-state fMRI data, one can assume that the signal ($\bx_{is}$) follows a normal distribution with covariance matrix $\bDelta_{i}$. While for task-based fMRI data ($\by_{it}$), this normality assumption does not hold as the signals are a convolution of the task onsite and the hemodynamic response function (HRF), where a widely accepted theoretical distribution for the HRF is a gamma distribution~\citep{lindquist2008statistical}. Thus, the solution to the following optimization problem is introduced as the estimator without imposing any distributional assumption.
\begin{eqnarray}\label{eq:obj}
    \underset{(\bgamma,\btheta,\alpha,\bbeta)}{\text{minimize}} && \ell=\frac{1}{n}\sum_{i=1}^{n}\left\{\log\left(\bgamma^\top\hat{\bSigma}_{i}\bgamma\right)-\alpha\log\left(\btheta^\top\hat{\bDelta}_{i}\btheta\right)-\bw_{i}^\top\bbeta\right\}^{2}, \nonumber \\
    \text{such that} && \bgamma^\top\bH_{y}\bgamma=1, \quad \btheta^\top\bH_{x}\btheta=1.
\end{eqnarray} 
$\hat{\bSigma}_{i}$ and $\hat{\bDelta}_{i}$ are estimators of $\bSigma_{i}$ and $\bDelta_{i}$, respectively. 
A natural choice is the sample covariance matrices, $\hat{\bSigma}_{i}=\bS_{i}^{y}=v_{i}^{-1}\sum_{t=1}^{v_{i}}(\by_{it}-\bar{\by}_{i})(\by_{it}-\bar{\by}_{i})^\top$ and $\hat{\bDelta}_{i}=\bS_{i}^{x}=u_{i}^{-1}\sum_{s=1}^{u_{i}}(\bx_{is}-\bar{\bx}_{i})(\bx_{is}-\bar{\bx}_{i})^\top$ (where $\bar{\by}_{i}=v_{i}^{-1}\sum_{t=1}^{v_{i}}\by_{it}$ and $\bar{\bx}_{i}=u_{i}^{-1}\sum_{s=1}^{u_{i}}\bx_{is}$), when the data dimensions are not too large and both $\bS_{i}^{y}$ and $\bS_{i}^{x}$ are positive definite. 
In the optimization problem~\eqref{eq:obj}, constraints are imposed on $\bgamma$ and $\btheta$ as the solutions are zero vectors otherwise, where $\bH_{y}\in\mathbb{R}^{q\times q}$ and $\bH_{x}\in\mathbb{R}^{p\times p}$ are both positive definite. Examples of such matrices include the identity matrix and the overall sample covariance matrices, $\bH_{y}=\bar{\bS}_{y}=\sum_{i=1}^{n}v_{i}\bS_{i}^{y}/\sum_{i=1}^{n}v_{i}$ and $\bH_{x}=\bar{\bS}_{x}=\sum_{i=1}^{n}u_{i}\bS_{i}^{x}/\sum_{i=1}^{n}u_{i}$. 
For a likelihood-based estimator, $\bH_{y}=\bar{\bS}_{y}$ and $\bH_{x}=\bar{\bS}_{x}$ are considered to incorporate sample information and avoid undesired component estimate~\cite[see][for a discussion]{zhao2021covariate}. For the proposed OLS estimator (or moment estimator), identity matrices, that is $\bH_{y}=\boldsymbol{\mathrm{I}}_{q}$ and $\bH_{x}=\boldsymbol{\mathrm{I}}_{p}$, are considered to relax the distributional constraints, analogous to the CCA.

Algorithm~\ref{alg:obj_solve} summarizes the estimation procedure of optimizing~\eqref{eq:obj}. One can show that the objective function in~\eqref{eq:obj} is bi-convex over $(\bgamma,\btheta,\alpha,\bbeta)$. Thus, a coordinate-descent algorithm is considered.
For model coefficients $\alpha$ and $\bbeta$, given the rest parameters, the updates are provided in the following.
\begin{equation}\label{eq:alpha_update}
    \hat{\alpha}=\left\{\frac{1}{n}\sum_{i=1}^{n}\log^{2}(\btheta^\top\hat{\bDelta}_{i}\btheta)\right\}^{-1}\left[\frac{1}{n}\sum_{i=1}^{n}\left\{\log(\bgamma^\top\hat{\bSigma}_{i}\bgamma)-\bw_{i}^\top\bbeta\right\}\log(\btheta^\top\hat{\bDelta}_{i}\btheta)\right],
\end{equation}
\begin{equation}\label{eq:beta_update}
    \hat{\bbeta}=\left(\frac{1}{n}\sum_{i=1}^{n}\bw_{i}\bw_{i}^\top\right)^{-1}\left[\frac{1}{n}\sum_{i=1}^{n}\left\{\log(\bgamma^\top\hat{\bSigma}_{i}\bgamma)-\alpha\log(\btheta^\top\hat{\bDelta}_{i}\btheta)\right\}\bw_{i}\right].
\end{equation}
For $\bgamma$ and $\btheta$, with quadratic constraints, the method of Lagrange multiplier is employed. Using $\bgamma$ as an example, let $U_{i}=\alpha\log(\btheta^\top\hat{\bDelta}_{i}\btheta)+\bw_{i}^\top\bbeta$, the Lagrangian form is
\begin{equation}\label{eq:gamma_Lag}
    \mathcal{L}(\bgamma)=\frac{1}{n}\sum_{i=1}^{n}\left\{\log(\bgamma^\top\hat{\bSigma}_{i}\bgamma)-U_{i}\right\}^{2}-\lambda_{1}(\bgamma^\top\bH_{y}\bgamma-1),
\end{equation}
where $\lambda_{1}$ is the Lagrange multiplier. Details of optimizing~\eqref{eq:gamma_Lag}, as well as optimizing for $\btheta$, are provided in Section~\ref{appendix:sub:obj_solve} of the supplementary materials.

\begin{algorithm}
    \caption{\label{alg:obj_solve}An algorithm of solving~\eqref{eq:obj}.}
    \begin{algorithmic}[1]
        \INPUT $\{(\by_{i1},\dots,\by_{iv_{i}}),(\bx_{i1},\dots,\bx_{iu_{i}}),\bw_{i}\}$

        \State For $i=1,\dots,n$, estimate $\bSigma_{i}$ and $\bDelta_{i}$, denoted as $\hat{\bSigma}_{i}$ and $\hat{\bDelta}_{i}$, respectively.

        \State \textbf{Initialization}: $(\bgamma^{(0)},\btheta^{(0)},\alpha^{(0)},\bbeta^{(0)})$

        \Repeat \; for iteration $h=0,1,2,\dots$,

            \State \; update $\alpha$ as $\alpha^{(h+1)}$ using~\eqref{eq:alpha_update} with $\left(\bgamma^{(h)},\btheta^{(h)},\bbeta^{(h)}\right)$;

            \State \; update $\bbeta$ as $\bbeta^{(h+1)}$ using~\eqref{eq:beta_update} with $\left(\bgamma^{(h)},\btheta^{(h)},\alpha^{(h+1)}\right)$;

            \State \; update $\btheta$ as $\btheta^{(h+1)}$ following the approach described in Section~\ref{appendix:sub:obj_solve} of the supplementary materials with $\left(\bgamma^{(h)},\alpha^{(h+1)},\bbeta^{(h+1)}\right)$;

            \State \; update $\bgamma$ as $\bgamma^{(h+1)}$ following the approach described in Section~\ref{appendix:sub:obj_solve} of the supplementary materials with $\left(\btheta^{(h+1)},\alpha^{(h+1)},\bbeta^{(h+1)}\right)$;

        \Until{the objective function in~\eqref{eq:obj} converges.}

        \State Consider a random series of initializations, repeat Steps 2--8, and choose the solution with the minimum objective value.

        \OUTPUT $(\hat{\bgamma},\hat{\btheta},\hat{\alpha},\hat{\bbeta})$
    \end{algorithmic}
\end{algorithm}
%------------------------------------------

%------------------------------------------
\subsection{Inference}
\label{sub:inference}

In this section, we focus on introducing a bootstrap procedure to perform inference on $\alpha$ and $\bbeta$, the model coefficients. 
Bootstrap inference on $\bgamma$ and $\btheta$ is beyond the scope of the current study as it may require a procedure of matching on the projections across bootstrap samples. In addition, the performance of this matching procedure highly depends the chosen similarity metric. Thus, a more rigorous and thorough theoretical and numerical investigation is left for future research.
The following is a procedure to conduct bootstrap inference on $\alpha$ and $\bbeta$.
\begin{description}
    \item[Step 0.] Obtain an estimate of $(\bgamma,\btheta)$, denoted as $(\hat{\bgamma},\hat{\btheta})$, using the full dataset.
    \item[Step 1.] Generate a bootstrap sample, $\{\{\by_{it}\}_{t}^{*},\{\bx_{is}\}_{s}^{*},\bw_{i}^{*}\}$, of size $n$ by sampling with replacement.
    \item[Step 2.] Estimate $(\alpha,\bbeta)$ using Algorithm~\ref{alg:obj_solve}  with $(\hat{\bgamma},\hat{\btheta})$ known.
    \item[Step 3.] Repeat Steps 1--2 for $B$ times.
    \item[Step 4.] Construct bootstrap confidence intervals for $(\alpha,\bbeta)$ under a prespecified significance level. 
\end{description}
In Step 1, the resampling procedure is conducted at the subject level. All observations within a subject will be used for estimation if being sampled. For datasets with a hierarchically nested structure, resampling on the highest level has been demonstrated to better preserve the original sample information and yield better performance~\citep{ren2010nonparametric}.
%------------------------------------------

%------------------------------------------
\subsection{Asymptotic properties}
\label{sub:asmp}

This section discusses the asymptotic properties of the proposed estimator under regularity conditions.
For $i=1,\dots,n$, it is assumed that $\bSigma_{i}$ has the eigendecomposition of $\bSigma_{i}=\bPi_{i}\bLambda_{i}\bPi_{i}^\top$ and $\bDelta_{i}$ has the eigendecomposition of $\bDelta_{i}=\bUpsilon_{i}\bOmega_{i}\bUpsilon_{i}^\top$, where $\bPi_{i}=(\bpi_{i1},\dots,\bpi_{iq})\in\mathbb{R}^{q\times q}$ and $\bUpsilon=(\bupsilon_{i1},\dots,\bupsilon_{ip})\in\mathbb{R}^{p\times p}$ are orthonormal eigenmatrices, $\bLambda_{i}=\mathrm{diag}\{\lambda_{i1},\dots,\lambda_{iq}\}\in\mathbb{R}^{q\times q}$ and $\bOmega=\mathrm{diag}\{\omega_{i1},\dots,\omega_{ip}\}\in\mathbb{R}^{p\times p}$ are diagonal matrices of corresponding eigenvalues. Let $\bzeta_{it}=\bPi_{i}^\top\by_{it}=(\zeta_{itk})_{k}\in\mathbb{R}^{q}$ and $\bkappa_{is}=\bUpsilon_{i}^\top\bx_{is}=(\kappa_{isj})_{j}\in\mathbb{R}^{p}$. Then, $\mathrm{Cov}(\bzeta_{it})=\bLambda_{i}$ and $\mathrm{Cov}(\bkappa_{is})=\bOmega_{i}$. The elements in $\bzeta_{it}$ are uncorrelated and so as the elements in $\bkappa_{is}$. The following assumptions are imposed.

\begin{description}
    \item[Assumption A1] Let $u=\min_{i}u_{i}$ and $v=\min_{i}v_{i}$. $p\ll u$ and $q\ll v$ are fixed.
    \item[Assumption A2] There exist constants $C_{1}$ independent of $u$ and $C_{2}$ independent of $v$, such that for $\forall~j=1,\dots,p$, $\mathbb{E}(\kappa_{i1j}^{4})\leq C_{1}$, and for $\forall~k=1,\dots,q$, $\mathbb{E}(\zeta_{i1k}^{4})\leq C_{2}$, for $\forall~i=1,\dots,n$.
    \item[Assumption A3] $\bSigma_{i}$'s share the same set of eigenvectors, i.e., $\bPi_{i}=\bPi$, for $i=1,\dots,n$; and $\bDelta_{i}$'s share the same set of eigenvectors, i.e., $\bUpsilon_{i}=\bUpsilon$, for $i=1,\dots,n$.
    \item[Assumption A4] For $\forall~i=1,\dots,n$, there exists (at least) a pair of columns in $\bPi_{i}$ and $\bUpsilon_{i}$, indexed by $k_{i}$ and $j_{i})$, respectively, such that $\bgamma=\bpi_{ik_{i}}$, $\btheta=\bupsilon_{ij_{i}}$, and Model~\eqref{eq:model} is satisfied.
\end{description}

Assumption A1 assumes a low-dimensional scenario for both $\bx$ and $\by$. Under this assumption, the sample covariance matrices are well-conditioned and are consistent estimators of the covariance matrices. Though distributional assumptions are not imposed, Assumption A2 regulates the higher-order moments on the transformed random variables, $\bkappa$ and $\bzeta$, which is equivalent to regulating the higher-order moments on $\bx$ and $\by$, respectively. Assumption A3 assumes the common diagonalization of $\bSigma_{i}$'s and $\bDelta_{i}$'s. Assumption A4 assumes the log-linear regression model to be correctly specified. With Assumptions A1--A4 satisfied, one can consider setting the eigenvectors of $\bar{\bS}_{y}$ and $\bar{\bS}_{x}$ as the initial value of $\bgamma$ and $\btheta$, respectively, in Algorithm~\ref{alg:obj_solve}. The following proposition suggests the consistency of the proposed estimator under such initialization.

\begin{proposition}\label{prop:asmp}
    Assume Assumptions A1--A4 hold. As $n,u,v\rightarrow\infty$, the estimator of $(\bgamma,\btheta,\alpha,\bbeta)$ obtained by Algorithm~\ref{alg:obj_solve} is asymptotically consistent.
\end{proposition}
%~~~~~~~~~~~~~~~~~~~~~~~~~~~~~~~~~~~~

%~~~~~~~~~~~~~~~~~~~~~~~~~~~~~~~~~~~~
%------------------------------------------

%------------------------------------------
\subsection{Higher-order components and choosing the number of components}
\label{sub:HO_comp}

Section~\ref{sub:estimation} introduces an algorithm to identify the first component using an OLS criterion. In this section, an approach is proposed to identify higher-order components. Here, the projections are assumed to be one-to-one, that is, one linear projection on $\by$ corresponds to one linear projection on $\bx$. Assume $\hat{\bGamma}^{(k-1)}=(\hat{\bgamma}_{1},\dots,\hat{\bgamma}_{(k-1)})\in\mathbb{R}^{q\times (k-1)}$ and $\hat{\bTheta}^{(k-1)}=(\hat{\btheta}_{1},\dots,\hat{\btheta}_{(k-1)})\in\mathbb{R}^{p\times (k-1)}$ are the first $(k-1)$ identified component pairs, for $k=2,\dots,\min(p,q)$. It is proposed to remove these identified components from the data to identify the next pair. Let
\begin{equation}
    \hat{\bY}_{i}^{(k)}=\bY_{i}-\bY_{i}\hat{\bGamma}^{(k-1)}\hat{\bGamma}^{(k-1)\top} \quad \text{and} \quad \hat{\bX}_{i}^{(k)}=\bX_{i}-\bX_{i}\hat{\bTheta}^{(k-1)}\hat{\bTheta}^{(k-1)\top},
\end{equation}
where $\bY_{i}=(\by_{i1},\dots,\by_{iv_{i}})^\top\in\mathbb{R}^{v_{i}\times q}$ and $\bX_{i}=(\bx_{i1},\dots,\bx_{iu_{i}})^\top\in\mathbb{R}^{u_{i}\times p}$ denote the data from subject $i$, for $i=1,\dots,n$. Consider $\{\hat{\bY}_{i}^{(k)},\hat{\bX}_{i}^{(k)},\bw_{i}\}$ as the new data and apply Algorithm~\ref{alg:obj_solve} to identify the $k$th component.

To choose the number of components, \citet{zhao2021covariate} introduced a criterion named average deviation from diagonality (DfD) for the problem of principal regression for covariance matrix outcomes. This was introduced based on the nature that the linear projections are a common diagonalization of the covariance matrices. Extending to the covariance-on-covariance regression problem considered in this manuscript, the following metric is considered. Let $\hat{\bGamma}^{(k)}\in\mathbb{R}^{q\times k}$ and $\hat{\bTheta}^{(k)}\in\mathbb{R}^{p\times k}$ be the estimated first $k$ components, define
\begin{equation}
    \mathrm{DfD}(k)=\max\left\{\mathrm{DfD}(\hat{\bGamma}^{(k)}), \mathrm{DfD}(\hat{\bTheta}^{(k)})\right\},
\end{equation}
where
\begin{equation}
    \mathrm{DfD}(\hat{\bGamma}^{(k)})=\prod_{i=1}^{n}\nu\left(\hat{\bGamma}^{(k)\top}\hat{\bSigma}_{i}\hat{\bGamma}^{(k)}\right)^{v_{i}/\sum_{i}v_{i}}, \text{ and } \mathrm{DfD}(\hat{\bTheta}^{(k)})=\prod_{i=1}^{n}\nu\left(\hat{\bTheta}^{(k)\top}\hat{\bDelta}_{i}\hat{\bTheta}^{(k)}\right)^{u_{i}/\sum_{i}u_{i}}.
\end{equation}
For a square matrix $\bA$, $\nu(\bA)=\det\{\mathrm{diag}(\bA)\}/\det(\bA)$, where $\mathrm{diag}(\bA)$ is a diagonal matrix with the same diagonal elements as in $\bA$ and $\det(\bA)$ is the determinant of $\bA$. $\nu(\bA)\geq1$ and the equality holds if and only if $\bA$ is a diagonal matrix. Considering a threshold $\varrho$, for example, a threshold of $\varrho=2$ recommended in \citet{zhao2021covariate}, the number of components is chosen as
\begin{equation}
    \hat{k}=\max\left\{k:\mathrm{DfD}(k)\leq \varrho\right\}.
\end{equation}
%------------------------------------------

%~~~~~~~~~~~~~~~~~~~~~~~~~~~~~~~~~~~~

%%%%%%%%%%%%%%%%%%%%%%%%%%%%%%%%%%%%%%%%%%%%%%%%%%%%%%%%%%

%%%%%%%%%%%%%%%%%%%%%%%%%%%%%%%%%%%%%%%%%%%%%%%%%%%%%%%%%%
% Simulation
%%%%%%%%%%%%%%%%%%%%%%%%%%%%%%%%%%%%%%%%%%%%%%%%%%%%%%%%%%
\section{Simulation Study}
\label{sec:sim}

%------------------------------------------
In this section, the performance of the proposed approach is evaluated through simulation studies. As no existing approach was designed to perform regression with multiple covariance predictors and multiple covariance outcomes, an approach integrating the common PCA~\citep{flury1984common} and regression, named as \textbf{CPCA-Reg}, is considered as the competing method. This CPCA-Reg approach has three steps. (1) Perform common PCA on $\{\bX_{i}\}$ and $\{\bY_{i}\}$ separately to obtain an estimate of the eigenvectors and the corresponding eigenvalues, where each subject is considered as a group. (2) Choose the top components that account for over $85\%$ of the total data variation, where the total data variation is calculated across all subjects. (3) For each pair of the chosen component, perform a linear regression using Model~\eqref{eq:model} to obtain an estimate of the coefficients. For the proposed \textbf{C}ovariance-\textbf{o}n-\textbf{C}ovariance \textbf{Reg}ression approach, denoted as \textbf{CoCReg}, sample covariance matrices are used to replace $\hat{\bSigma}_{i}$ and $\hat{\bDelta}_{i}$ in the algorithm.

In Simulation (i), common eigenstructure is assumed. Covariance matrices are generated following the eigendecompositions, $\bSigma_{i}=\bPi\bLambda_{i}\bPi^\top$ and $\bDelta_{i}=\bUpsilon\bOmega_{i}\bUpsilon^\top$, where $\bPi=(\bpi_{1},\dots,\bpi_{q})\in\mathbb{R}^{q\times q}$ and $\bUpsilon=(\bupsilon_{1},\dots,\bupsilon_{p})\in\mathbb{R}^{p\times p}$ are orthonormal matrices of common eigenvectors, $\bLambda_{i}=\mathrm{diag}\{\lambda_{i1},\dots,\lambda_{iq}\}\in\mathbb{R}^{q\times q}$ and $\bOmega_{i}=\mathrm{diag}\{\omega_{i1},\dots,\omega_{ip}\}\in\mathbb{R}^{p\times p}$ are diagonal matrices of individual eigenvalues, for $i=1,\dots,n$. 
Two pairs of components are chosen to satisfy Model~\eqref{eq:model}: C1 of $(\bpi_{2},\bupsilon_{1})$ and C2 of $(\bpi_{4},\bupsilon_{3})$. For $\bw_{i}$, a case of $r=2$ is considered, where the first element is one for the intercept term and the second element is generated from a Bernoulli distribution with probability $0.5$ of being one. In C1, $\alpha=3$ and $\bbeta=(1,-1)^\top$, and in C2, $\alpha=2$ and $\bbeta=(-1,1)^\top$.
For the rest dimensions, the eigenvalues are generated from a log-normal distribution with mean value decreasing from $1$ to $-2$ and standard deviation $0.1$. Thus, the diagonal elements in $\bLambda_{i}$ and $\bOmega_{i}$ are exponentially decaying. 
With the covariance matrices, $\by_{it}$ and $\bx_{is}$ are generated from the multivariate normal distribution with mean zero. Here, consider the number of observations within a subject is the same with $v_{i}=v$ and $u_{i}=u$ for $i=1,\dots,n$.
Two scenarios of data dimension and sample size are considered, $(p,q)=(10,5),~(n,u,v)=(100,100,100)$ and $(p,q)=(100,100),~(n,u,v)=(500,500,500)$, where the data dimension in the second scenario is close to the HCP-A data application in Section~\ref{sec:fmri}. For the larger dimension case, only one component (C1) is assumed to satisfy the model assumption to have distinguishable eigenvalues.
In Simulation (ii), a setting of partial common diagonalization is considered for the dimension of $(p,q)=(10,5)$. For $\bDelta_{i}$'s, the first five eigenvectors are assumed to be identical across subjects; and for $\bSigma_{i}$'s, the first three eigenvectors are identical. The rest eigenvectors are randomly generated for each subject but to satisfy the orthonormal condition. The rest parameter settings are the same as in Simulation (i). Thus, in this simulation, C1 still satisfies the assumptions; while in C2, only estimating $\btheta$ satisfies the assumptions, but estimating $\bgamma$ does not.
To evaluate the performance of identifying $\bgamma$ and $\btheta$, the absolute value of the inner product of the estimate and the truth, that is $|\langle\hat{\bgamma},\bgamma\rangle|$ and $|\langle\hat{\btheta},\btheta\rangle|$, is used as a similarity metric between two unit-norm vectors.

Table~\ref{table:sim_est} presents the results. 
In Simulation (i), for the case of $(p,q)=(10,5)$ and $(n,u,v)=(100,100,100)$, the proposed CoCReg approach identifies the two components with high similarities, where the similarity of estimating $\bgamma$ is over $0.980$ and the similarity of estimating $\btheta$ is over $0.960$. The competing method of CPCA-Reg only identifies the first component (C1), as the first component may already account for over $85\%$ of the data variation. Though the CPCA-Reg approach yields a higher similarity in estimating $\bgamma$ and $\btheta$, the bias of estimating $\alpha$ is higher using a second-step regression model. Here, a higher similarity in estimating $\bgamma$ and $\btheta$ from the CPCA-Reg approach is expected as the data are generated from multivariate normal distributions under the common PCA assumption. The likelihood-based estimator of $\bgamma$ and $\btheta$ from CPCA-Reg yields the optimal performance. However, the CPCA-Reg approach requires a pair-matching procedure when fitting the regression models. Rather, the proposed CoCReg approach directly identifies the target projections.
For the scenario of $(p,q)=(100,100)$ and $(n,u,v)=(500,500,500)$, the CoCReg approach yields a good estimate in both $\bgamma$ and $\btheta$. The bias of estimating $\alpha$ and $\bbeta$, especially $\alpha$, is lower, compared to the CPCA-Reg approach. When the sample size increases to $(n,u,v)=(1000,1000,1000)$, the performance of CoCReg improves with lower bias, standard error (SE), and mean squared error (MSE) in estimating the parameters. However, the CPCA-Reg approach fails to execute as it requires large computing memory for such a large sample size.
In Simulation (ii), the performance of the proposed approach in identifying C1 is very close to the performance in Simulation (i). For C2, though the common diagonalization assumption does not hold for $\bgamma$, the proposed approach correctly identifies $\btheta$ while the bias of estimating model coefficients is higher. Again, the CPCA-Reg approach identifies C1 with higher bias in estimating $\alpha$ and fails to identify C2.
Section~\ref{appendix:sub:sim_nonGaussian} of the supplementary materials presents a simulation study for non-Gaussian distributed data. The results demonstrate the robustness of the proposed estimator to both symmetric (multivariate $t$ with degrees of freedom $\nu=3$) and skewed (matrix gamma) distributions.

To evaluate the finite sample performance of the CoCReg approach, Figure~\ref{fig:sim_asmp} presents the performance on estimating the first component (C1) at various combinations of the sample sizes for the case of $p=10$ and $q=5$. From the figures, as both the sample size ($n$) and the number of observations within each subject ($u$ and $v$) increase, the estimate of the parameters converge to the truth and the SE and MSE converge to zero. Following the procedure introduced in Section~\ref{sub:inference}, $95\%$ confidence intervals of the model coefficients are constructed from $500$ bootstrap samples. As $n,u,v$ increase, the coverage probability (CP) converges to the designated level.
%------------------------------------------

%------------------------------------------
% simulation
% (i) p=10, q=5: 220510/case1
% (i) p=100, q=100: 220510/case2
% (ii) p=10, q=5: 220510/case3
\begin{table}
% \begin{sidewaystable}
    \caption{\label{table:sim_est}Performance in identifying target components and estimating model coefficients in the simulation study. SE: standard error; MSE: mean squared error.}
    \begin{center}
        \resizebox{\textwidth}{!}{
        \begin{tabular}{c c c l l r r r r r c r r r}
            \hline
            & & & & & & & \multicolumn{3}{c}{$\hat{\alpha}$} && \multicolumn{3}{c}{$\hat{\beta}_{1}$} \\
            \cline{8-10}\cline{12-14}
            & \multicolumn{1}{c}{\multirow{-2}{*}{$(p,q)$}} & \multicolumn{1}{c}{\multirow{-2}{*}{$(n,u,v)$}} & \multicolumn{1}{c}{\multirow{-2}{*}{Method}} & & \multicolumn{1}{c}{\multirow{-2}{*}{$|\langle\hat{\bgamma},\bgamma\rangle|$ (SE)}} & \multicolumn{1}{c}{\multirow{-2}{*}{$|\langle\hat{\btheta},\btheta\rangle|$ (SE)}} & \multicolumn{1}{c}{Bias} & \multicolumn{1}{c}{SE} & \multicolumn{1}{c}{MSE} && \multicolumn{1}{c}{Bias} & \multicolumn{1}{c}{SE} & \multicolumn{1}{c}{MSE} \\
            \hline
            & & & & C1 & $0.984$ ($0.062$) & $0.964$ ($0.032$) & $-0.131$ & $0.131$ & $0.034$ && $0.037$ & $0.092$ & $0.010$ \\
            & & & \multirow{-2}{*}{CoCReg} & C2 & $0.988$ ($0.039$) & $0.960$ ($0.047$) & $-0.142$ & $0.304$ & $0.112$ && $-0.100$ & $0.225$ & $0.060$ \\
            \cline{4-14}
            & & & & C1 & $0.999$ ($0.000$) & $0.999$ ($0.000$) & $-0.234$ & $0.088$ & $0.062$ && $-0.002$ & $0.087$ & $0.007$ \\
            & \multirow{-4}{*}{$(10,5)$} & \multirow{-4}{*}{($100,100,100$)} & \multirow{-2}{*}{CPCA-Reg} & C2 & \multicolumn{1}{c}{-} & \multicolumn{1}{c}{-} & \multicolumn{1}{c}{-} & \multicolumn{1}{c}{-} & \multicolumn{1}{c}{-} && \multicolumn{1}{c}{-} & \multicolumn{1}{c}{-} & \multicolumn{1}{c}{-} \\
            \cline{2-14}
            % & & & & C1 \\
            & & & \multirow{-1}{*}{CoCReg} & C1 & $0.999$ ($0.000$) & $0.996$ ($0.005$) & $-0.333$ & $0.039$ & $0.112$ && $-0.001$ & $0.017$ & $0.000$ \\
            \cline{4-14}
            % & & & & C1 \\
            & & \multirow{-2}{*}{($500,500,500$)} & \multirow{-1}{*}{CPCA-Reg} & C1 & $0.999$ ($0.000$) & $0.999$ ($0.000$) & $-0.860$ & $0.067$ & $0.745$ && $0.001$ & $0.015$ & $0.000$ \\
            \cline{3-14}
            & & & \multirow{-1}{*}{CoCReg} & C1 & $1.000$ ($0.000$) & $0.998$ ($0.002$) & $-0.158$ & $0.018$ & $0.025$ && $-0.000$ & $0.008$ & $0.000$ \\
            \cline{4-14}
            \multirow{-8}{*}{(i)} & \multirow{-4}{*}{$(100,100)$} & \multirow{-2}{*}{($1000,1000,1000$)} &  \multirow{-1}{*}{CPCA-Reg} & C1 & \multicolumn{1}{c}{-} & \multicolumn{1}{c}{-} & \multicolumn{1}{c}{-} & \multicolumn{1}{c}{-} & \multicolumn{1}{c}{-} && \multicolumn{1}{c}{-} & \multicolumn{1}{c}{-} & \multicolumn{1}{c}{-} \\
            \hline
            & & & & C1 & $0.985$ ($0.056$) & $0.964$ ($0.038$) & $-0.129$ & $0.117$ & $0.030$ && $0.025$ & $0.094$ & $0.009$ \\
            & & & \multirow{-2}{*}{CoCReg} & C2 & \multicolumn{1}{c}{-} & $0.959$ ($0.076$) & $-0.299$ & $0.379$ & $0.232$ && $-0.180$ & $0.261$ & $0.100$ \\
            \cline{4-14}
            & & & & C1 & $0.999$ ($0.000$) & $0.999$ ($0.000$) & $0.269$ & $0.108$ & $0.084$ && $-0.002$ & $0.087$ & $0.007$ \\
            \multirow{-4}{*}{(ii)} & \multirow{-4}{*}{$(10,5)$} & \multirow{-4}{*}{($100,100,100$)} & \multirow{-2}{*}{CPCA-Reg} & C2 & \multicolumn{1}{c}{-} & \multicolumn{1}{c}{-} & \multicolumn{1}{c}{-} & \multicolumn{1}{c}{-} & \multicolumn{1}{c}{-} && \multicolumn{1}{c}{-} & \multicolumn{1}{c}{-} & \multicolumn{1}{c}{-} \\
            \hline
        \end{tabular}
        }
    \end{center}
% \end{sidewaystable}
\end{table}
\begin{figure}
    \begin{center}
        \subfloat[$|\langle\hat{\bgamma},\bgamma\rangle|$]{\includegraphics[width=0.25\textwidth]{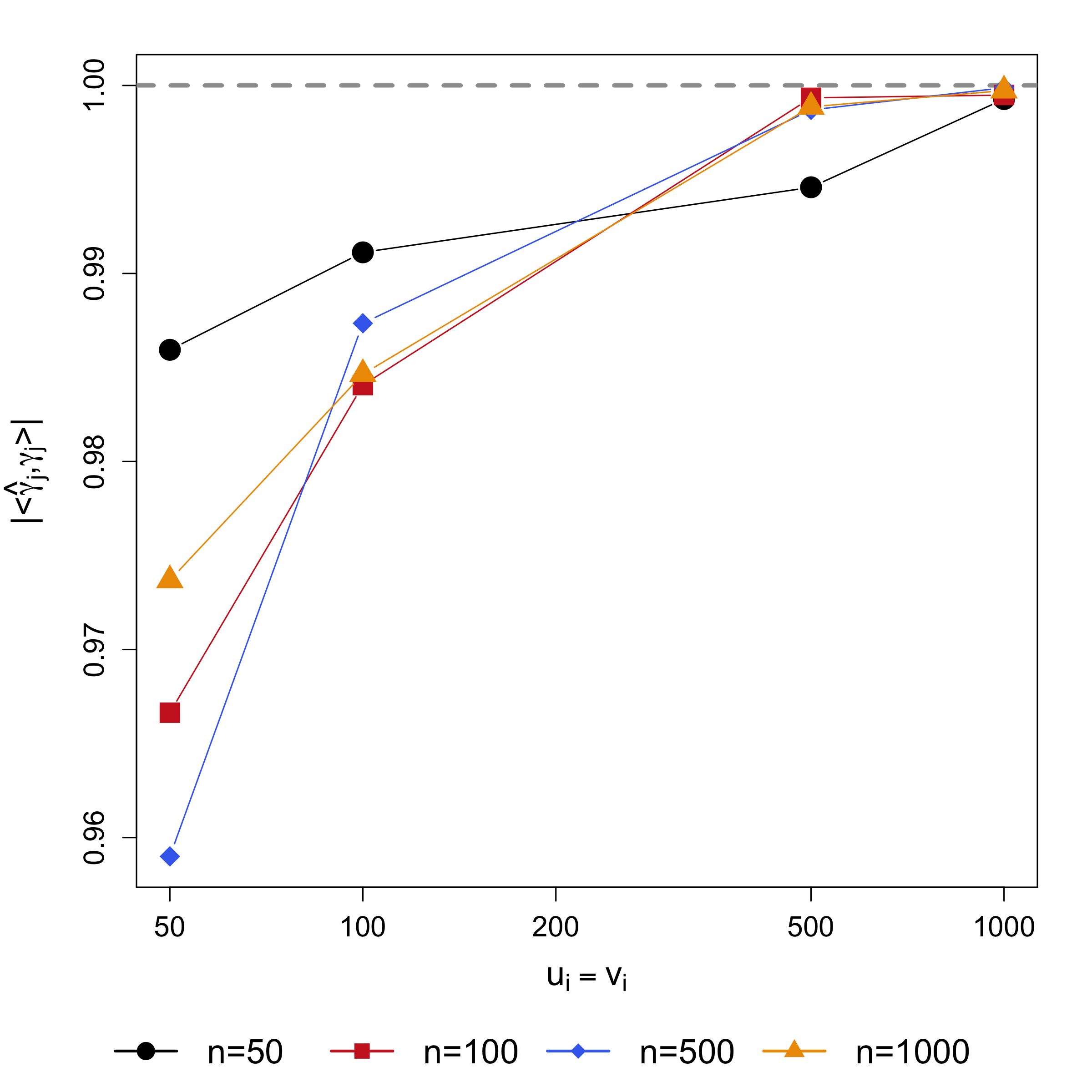}}
        \subfloat[$\mathrm{SE}(|\langle\hat{\bgamma},\bgamma\rangle|)$]{\includegraphics[width=0.25\textwidth]{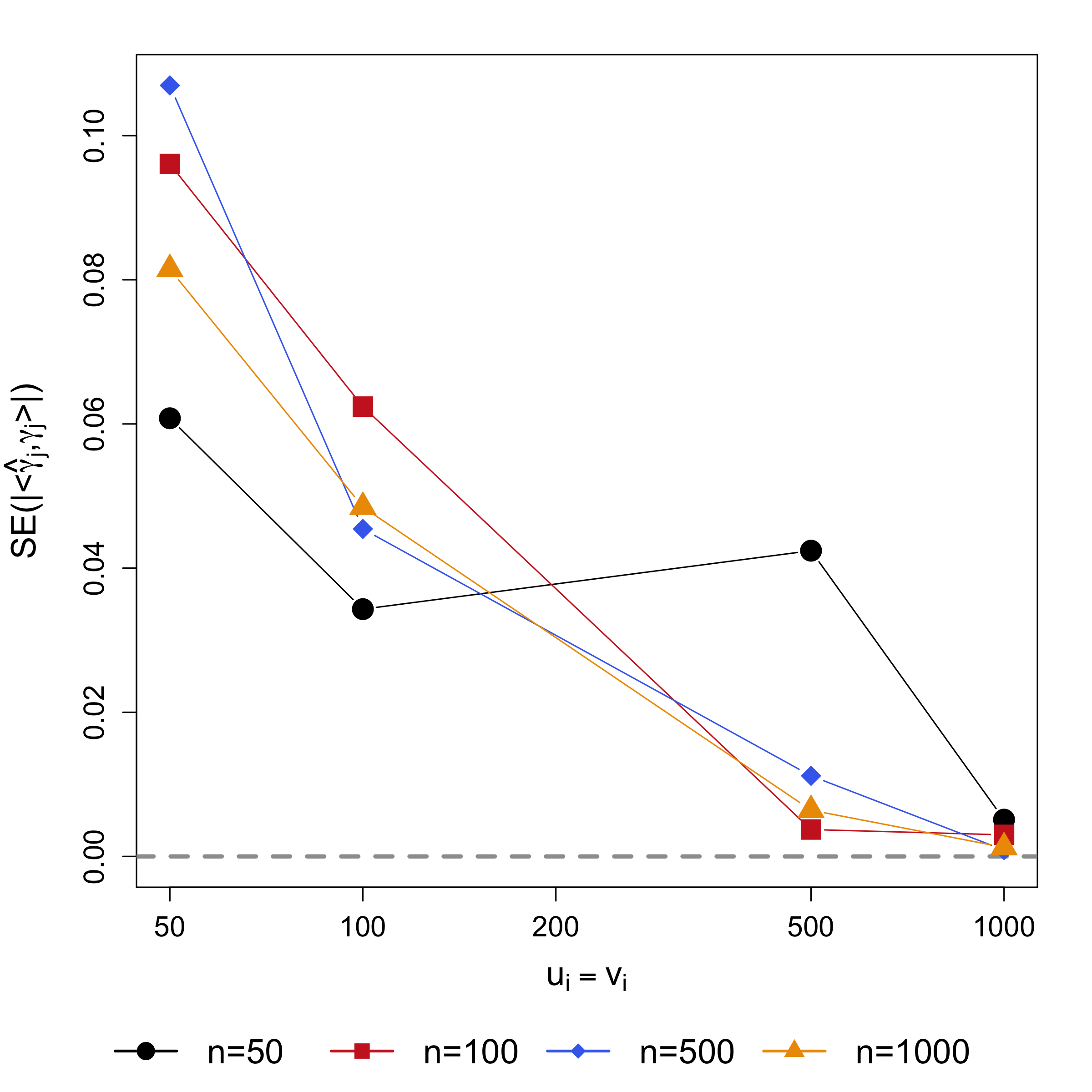}}
        \subfloat[$|\langle\hat{\btheta},\btheta\rangle|$]{\includegraphics[width=0.25\textwidth]{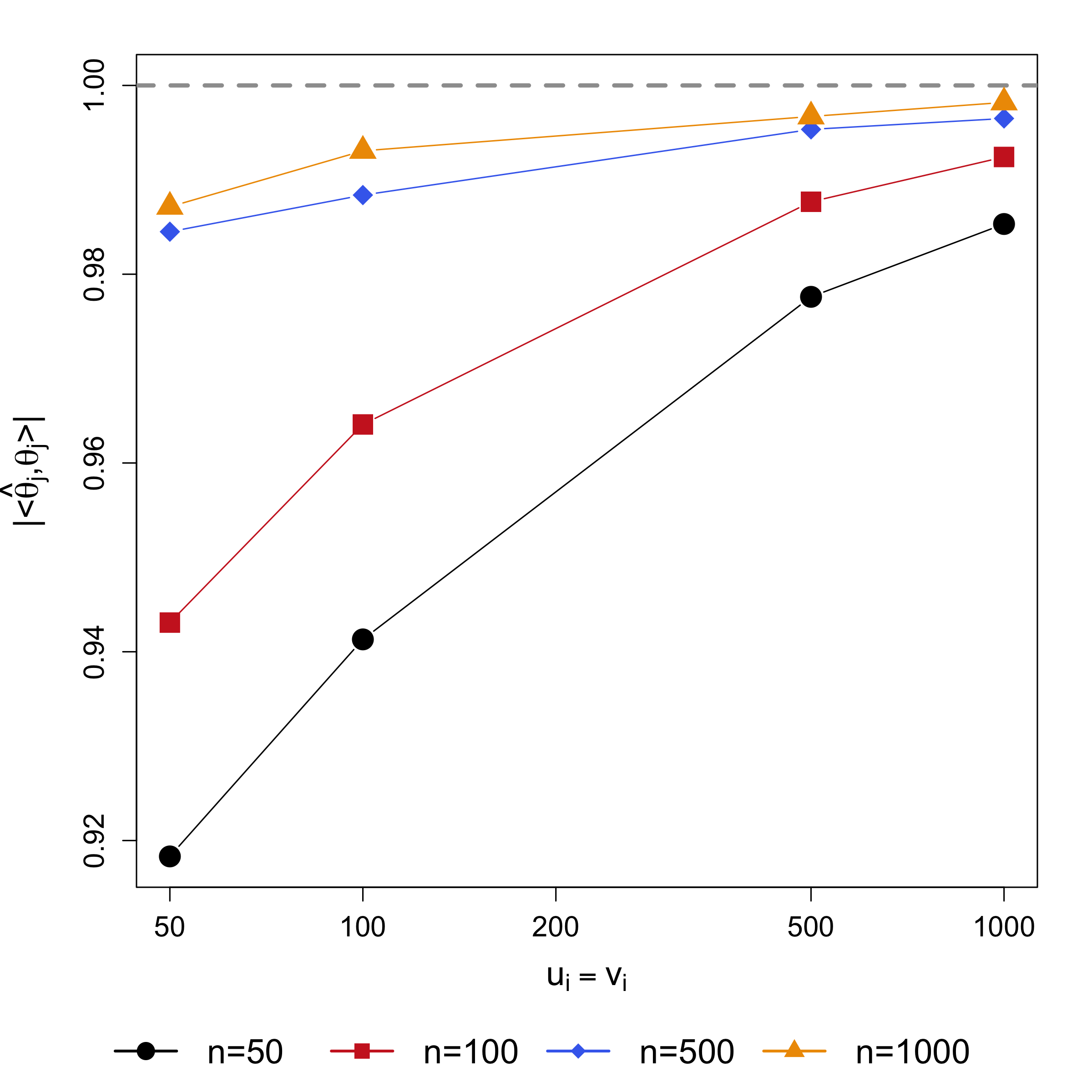}}
        \subfloat[$\mathrm{SE}(|\langle\hat{\btheta},\btheta\rangle|)$]{\includegraphics[width=0.25\textwidth]{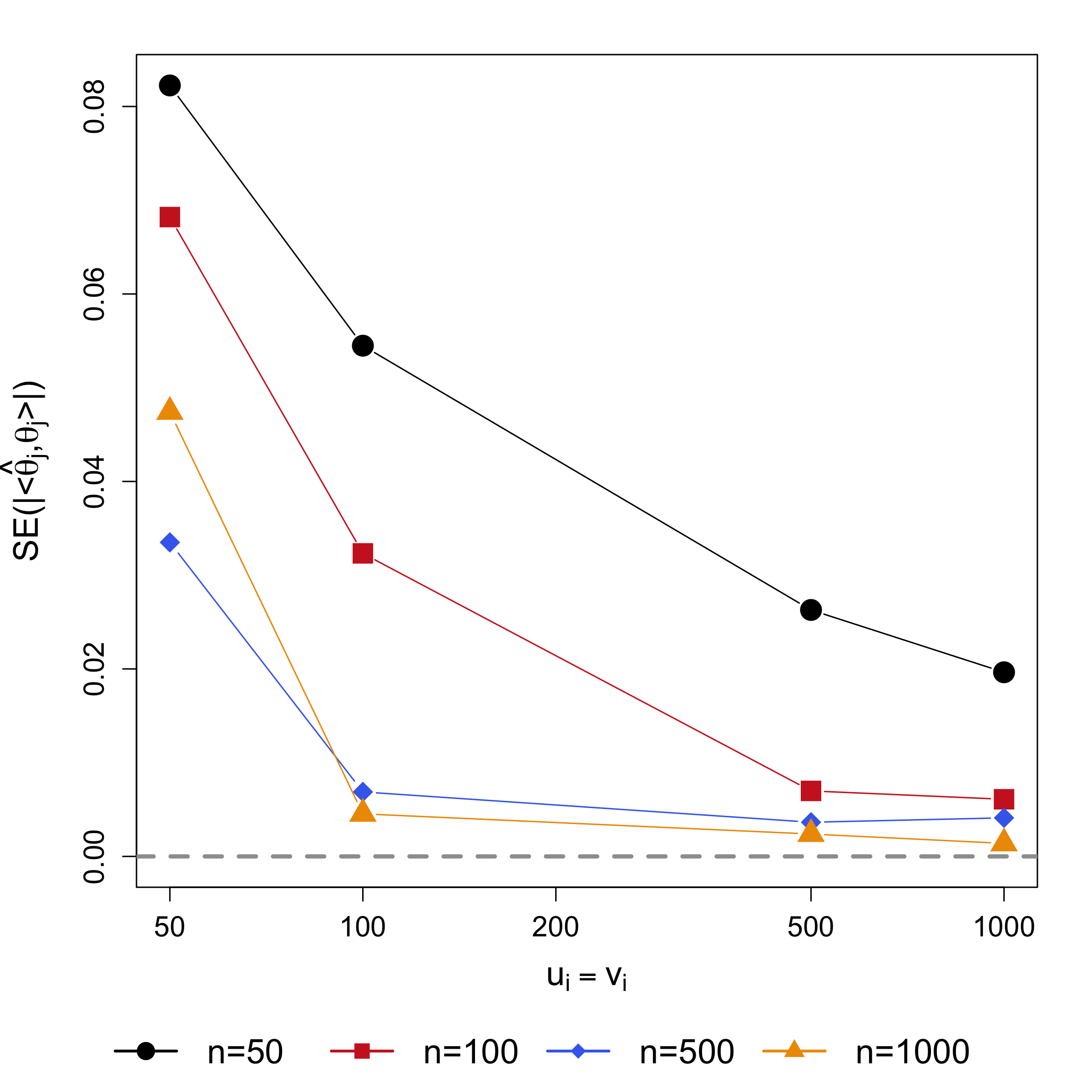}}

        \subfloat[$\hat{\alpha}$]{\includegraphics[width=0.25\textwidth]{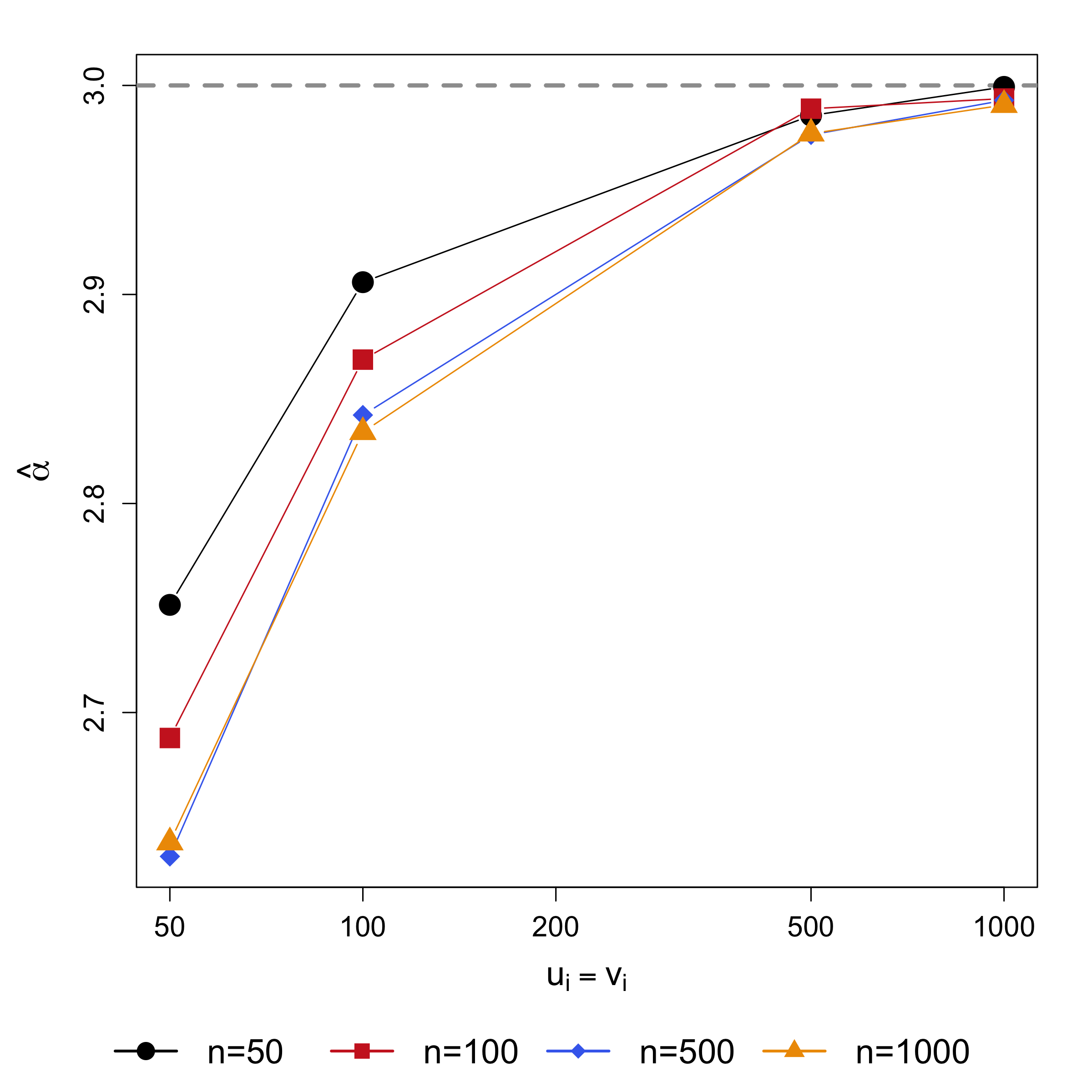}}
        \subfloat[$\mathrm{SE}(\hat{\alpha})$]{\includegraphics[width=0.25\textwidth]{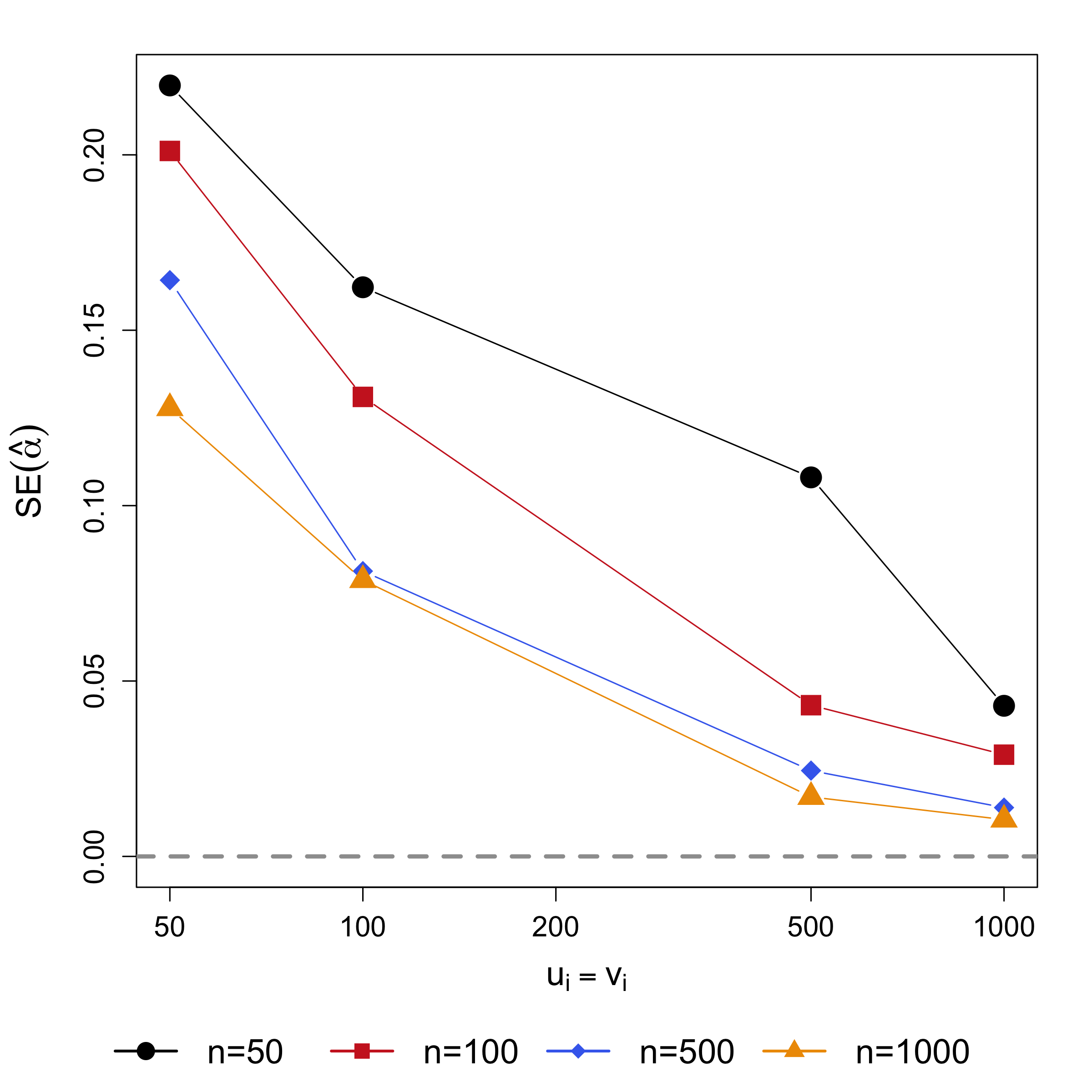}}
        \subfloat[$\mathrm{MSE}(\hat{\alpha})$]{\includegraphics[width=0.25\textwidth]{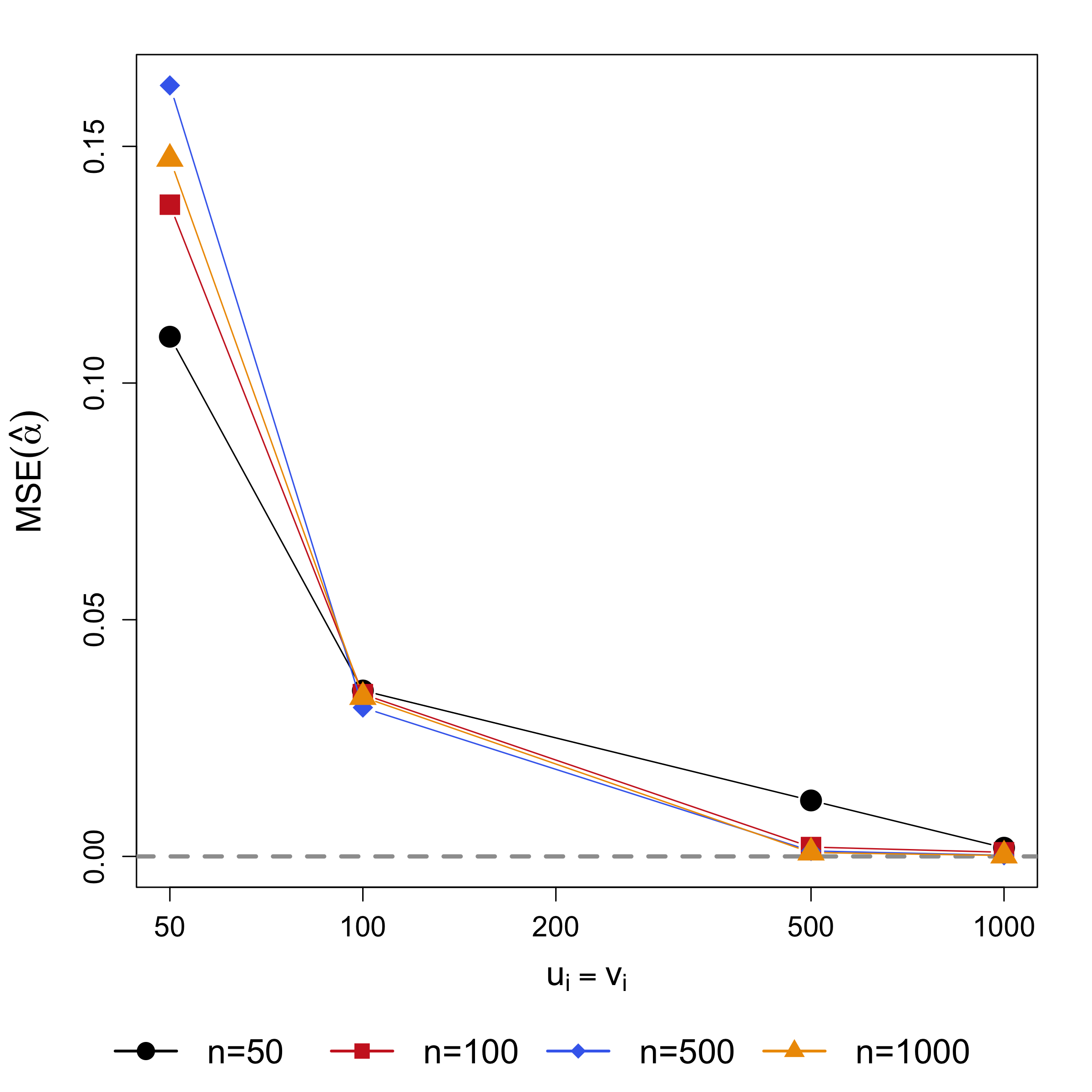}}
        \subfloat[$\mathrm{CP}(\hat{\alpha})$]{\includegraphics[width=0.25\textwidth]{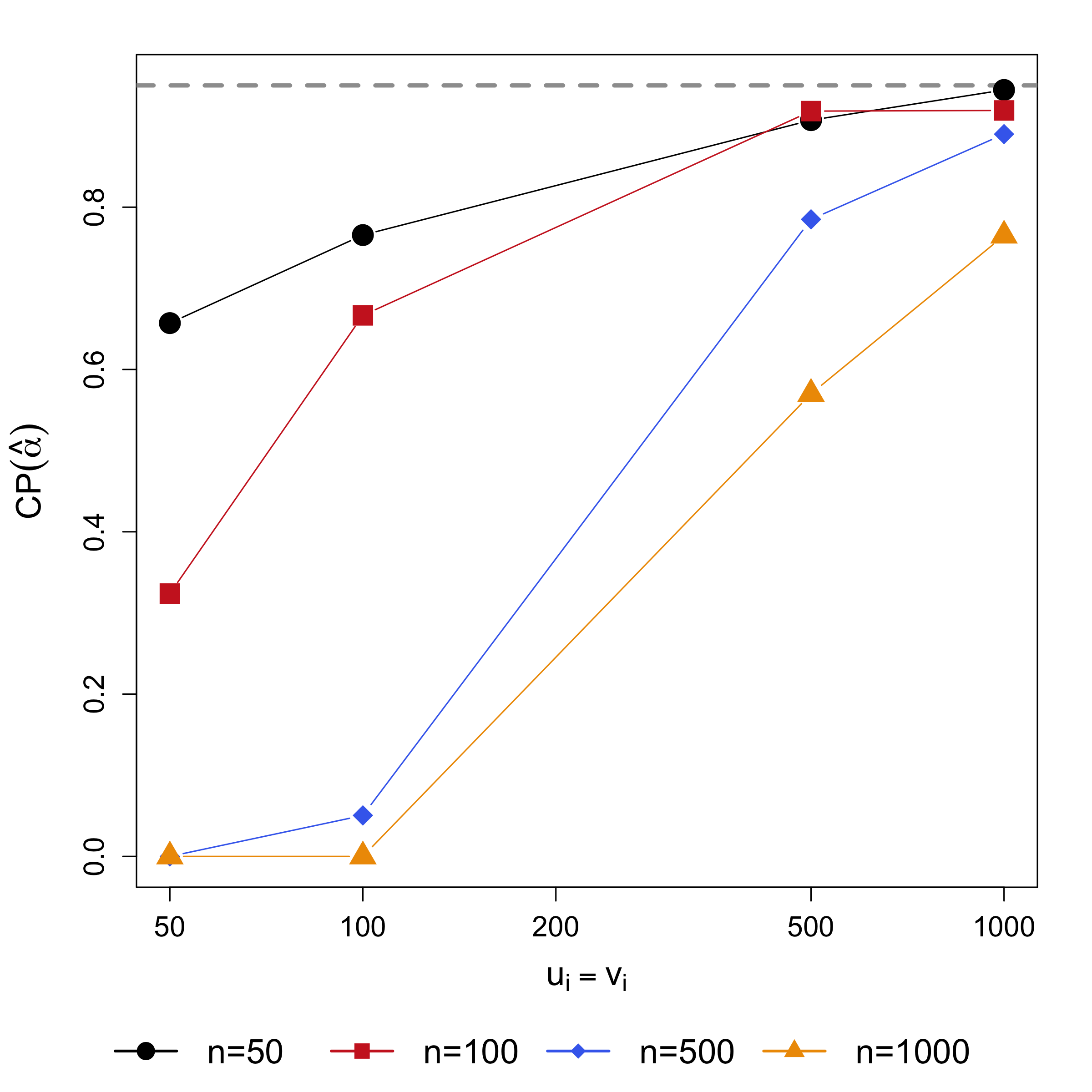}}

        \subfloat[$\hat{\beta}_{1}$]{\includegraphics[width=0.25\textwidth]{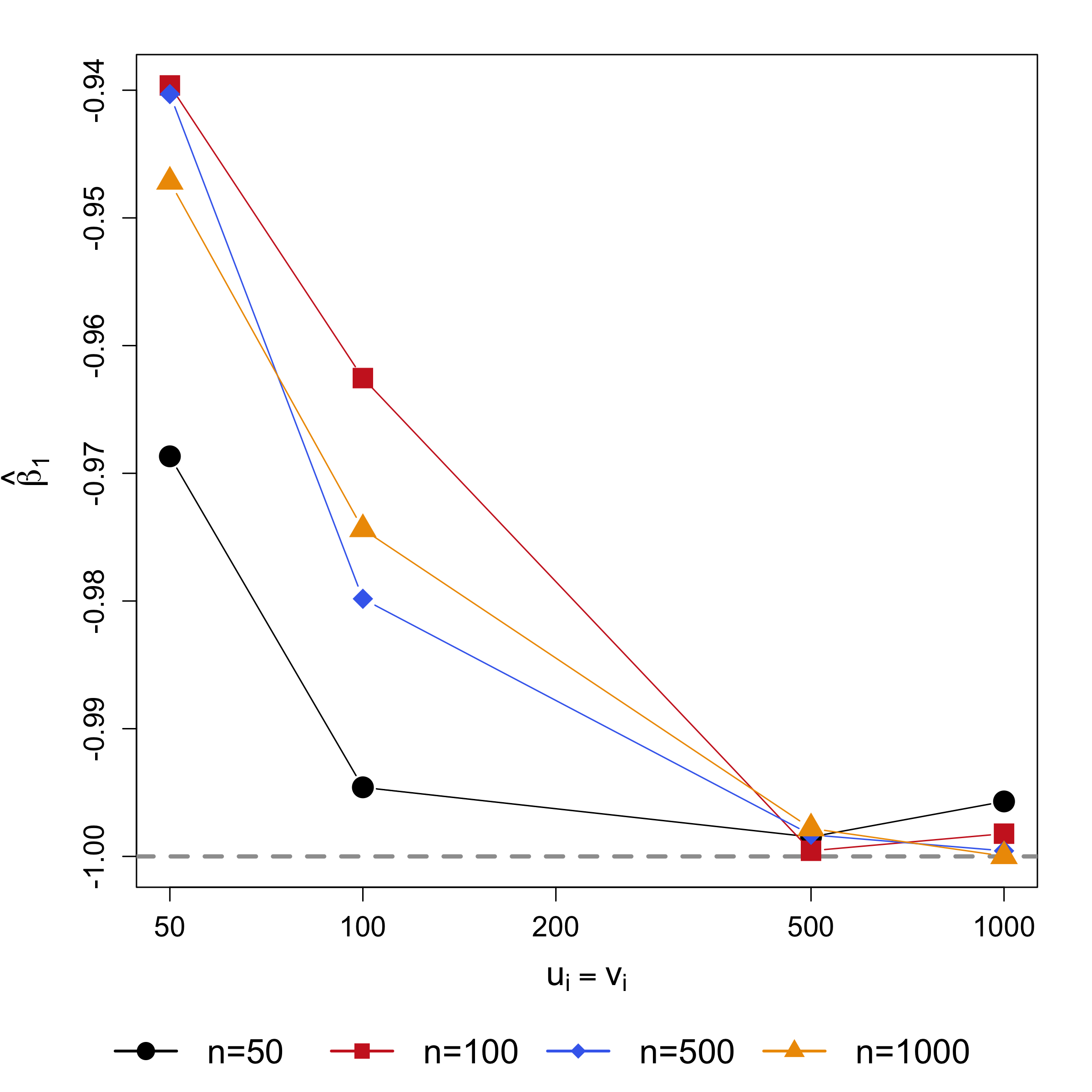}}
        \subfloat[$\mathrm{SE}(\hat{\beta}_{1})$]{\includegraphics[width=0.25\textwidth]{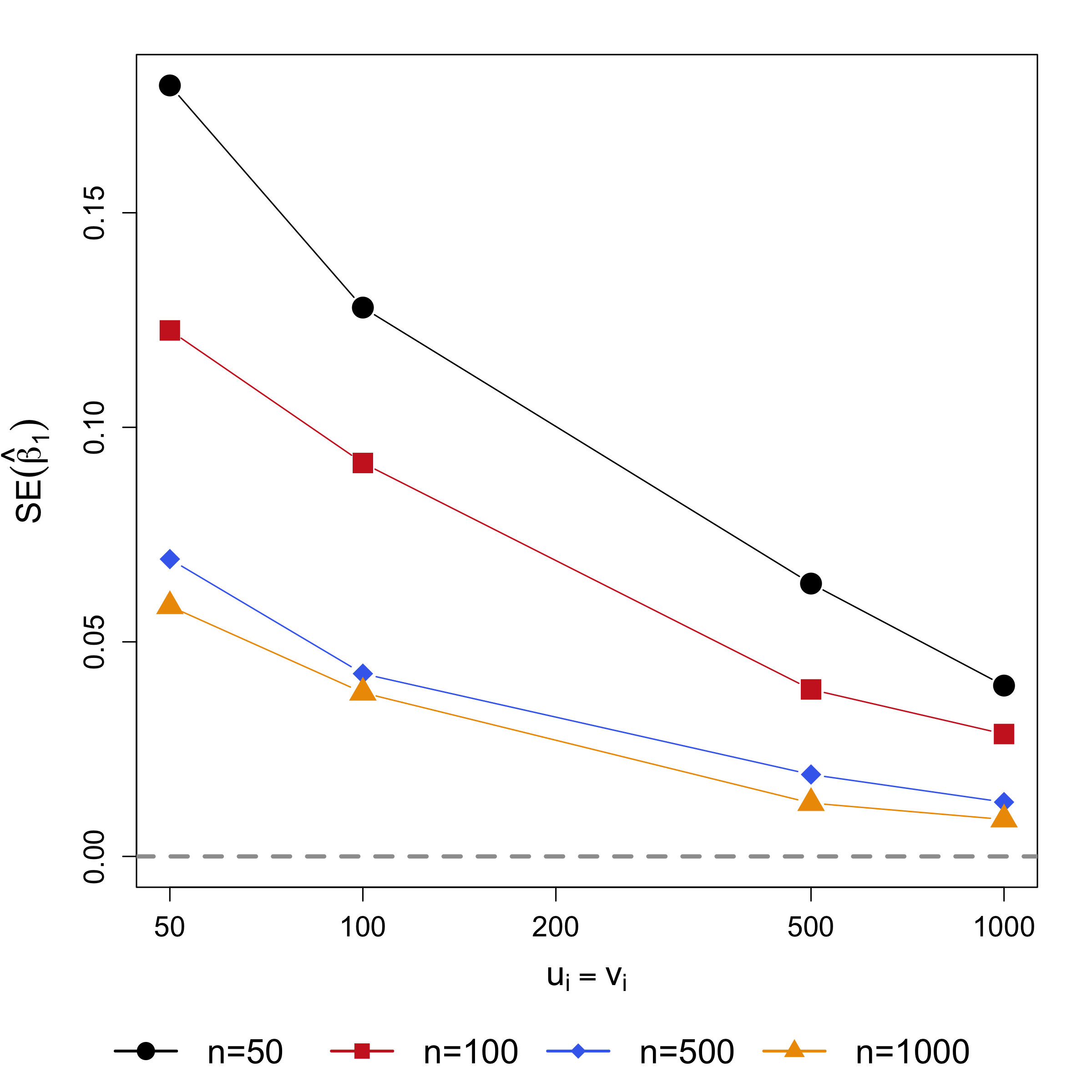}}
        \subfloat[$\mathrm{MSE}(\hat{\beta}_{1})$]{\includegraphics[width=0.25\textwidth]{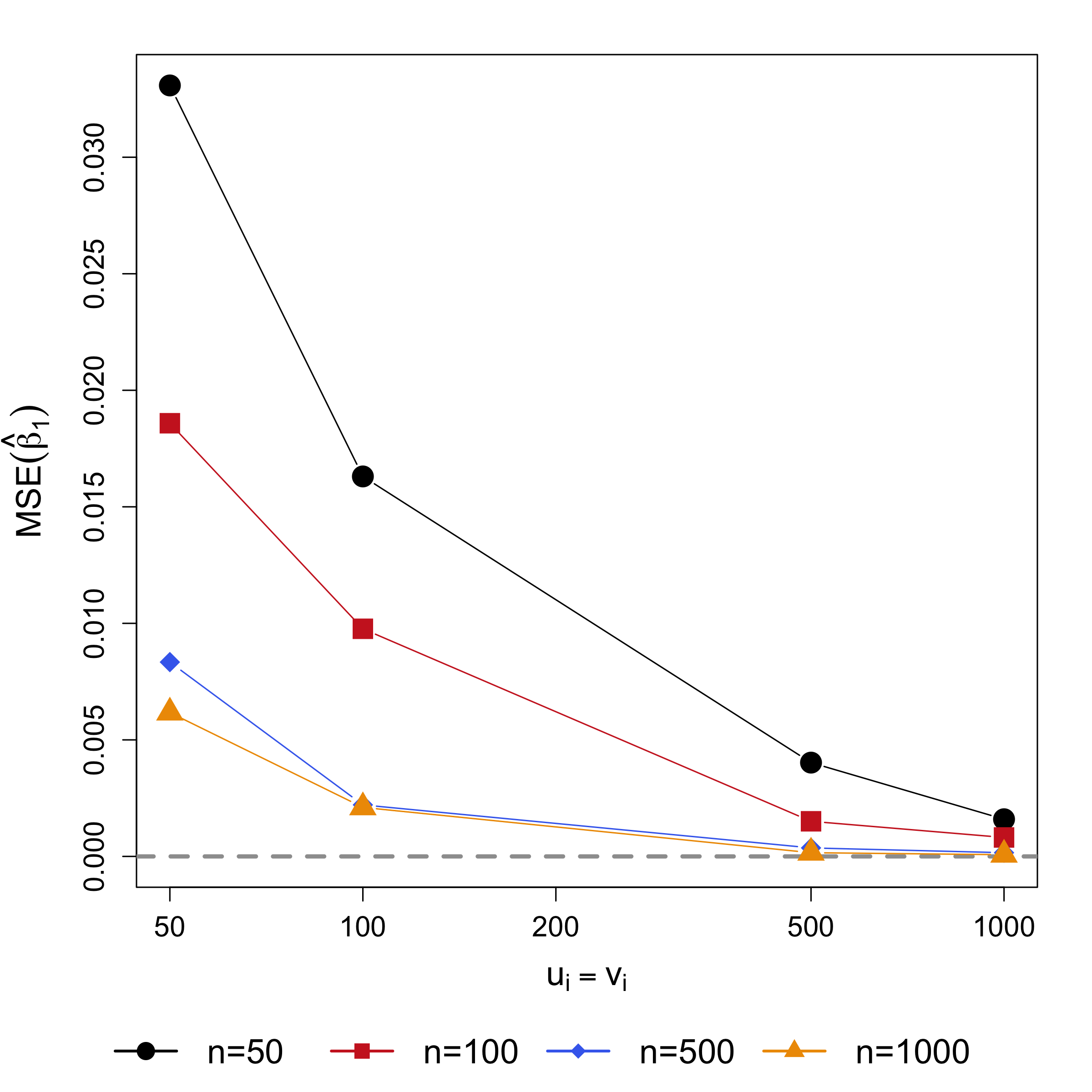}}
        \subfloat[$\mathrm{CP}(\hat{\beta}_{1})$]{\includegraphics[width=0.25\textwidth]{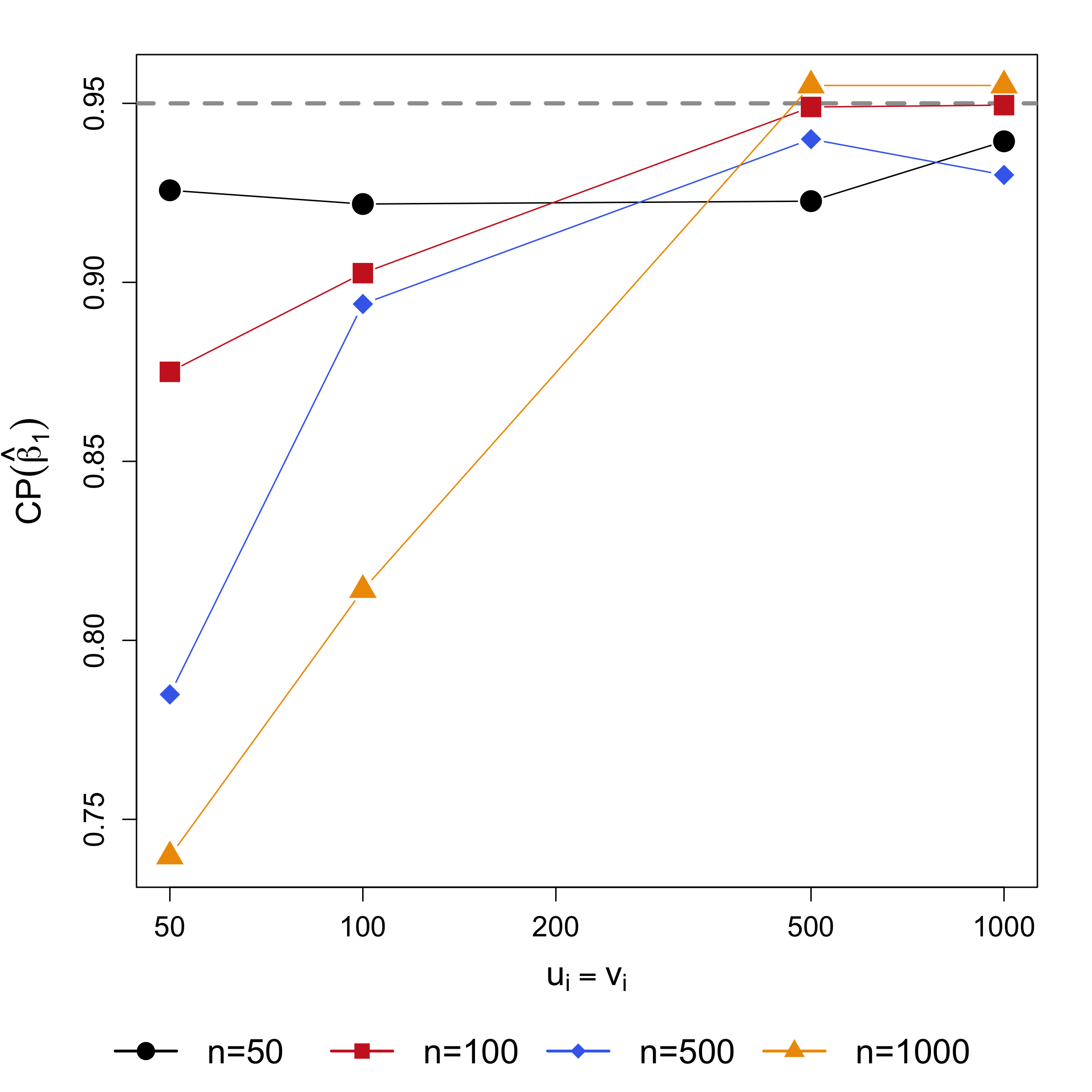}}
    \end{center}
    \caption{\label{fig:sim_asmp}Estimation performance of the proposed approach (CoCReg) on estimating the first component (C1) as the sample size $(n,u,v)$ varies with $p=10$ and $q=5$ in the simulation study. SE: standard error; MSE: mean squared error; CP: coverage probability.}
\end{figure}
%------------------------------------------
%%%%%%%%%%%%%%%%%%%%%%%%%%%%%%%%%%%%%%%%%%%%%%%%%%%%%%%%%%

%%%%%%%%%%%%%%%%%%%%%%%%%%%%%%%%%%%%%%%%%%%%%%%%%%%%%%%%%%
% Real data
%%%%%%%%%%%%%%%%%%%%%%%%%%%%%%%%%%%%%%%%%%%%%%%%%%%%%%%%%%
\section{The Lifespan Human Connectome Project Aging Study}
\label{sec:fmri}

%-------------------------------
We apply the proposed approach to the Lifespan Human Connectome Project (HCP) Aging Study. The goal of the aging study is to utilize the technological advances developed by the HCP study of healthy young adults to explore typical aging trajectory and how connectome changes in the brain among mature and older adults. In this study, it is considered to use resting-state functional connectivity to predict functional connectivity during a FACENAME task, a paired-associates memory task adapted from~\citet{zeineh2003dynamics}. During the task-based scanning, participants were instructed to encode the names of different face identities in some blocks and mentally recall them in other blocks. The task was designed as a localizer for encoding-related and retrieval-related neural processes. A total of $n=551$ subjects aged between $36$ and $90$ ($245$ Male and $306$ Female) having both the resting-state and task-based fMRI available without any quality control issue are included in the analysis. The fMRI data were minimally preprocessed~\citep{glasser2013minimal}. Signals are extracted from $p=q=75$ brain regions, including $60$ cortical and $15$ subcortical regions grouped into $10$ functional modules, using the Harvard-Oxford Atlas in FSL~\citep{smith2004advances}. Both resting-state and task-based signals are motion corrected. In the resting-state data ($\bx_{is}$), there are $u_{i}=u=478$ observations of each subject and in the task-based data ($\by_{it}$), the number is $v_{i}=v=335$. Other covariates include gender at birth and age ($\bw_{i}$), thus $r=3$ including the intercept.
The validity of model assumptions is examined and discussed in Section~\ref{appendix:subsec:fmri_assumption} of the supplementary materials.

Using the DfD criterion introduced in Section~\ref{sub:HO_comp}, the proposed approach identifies three orthogonal components, denoted as C1, C2, and C3. Table~\ref{tabel:hcp} presents the estimated model coefficients and $95\%$ bootstrap confidence intervals from $500$ samples. In all three components, functional connectivity within the identified task-based network is significantly associated with functional connectivity within the identified resting-state network, positive in C1 and C2 and negative in C3. To better interpret the components, $\bgamma$ and $\btheta$ are sparsified following an \textit{ad hoc} procedure using a fused lasso regression~\citep{tibshirani2005sparsity}, where local smoothness and constancy are imposed within each brain functional module~\citep{grosenick2013interpretable}. Figure~\ref{appendix:fig:hcp} in Section~\ref{appendix:subsec:fmri_results} of the supplementary materials presents the sparsified loading profile colored by the functional modules and Figure~\ref{fig:hcp} presents the nonzero loaded regions in a brain map and the river plot of the loading configuration by functional modules.

% C1
The resting-state network of C1 is primarily the default mode network (DMN), which is best known for being active during wakeful rest and consistently captured in rs-fMRI studies~\citep{biswal1995functional,shen2015resting}. The related task-state network is primarily cortical areas of the brain and the cerebellum with similar positive loading values suggesting a global signal pattern during the task. Emerging evidence suggests that there is a potential extrinsic mode network (EMN) during tasks that can be predicted by DMN at the resting state~\citep{hugdahl2015existence}. Our finding confirms this point with the identified task-state network signatures contributing to the establishment of EMN.
% C2
The two networks in C2 yields a high similarity of $0.663$ and the module configuration is consistent in the DMN, the fronto-parietal network, the limbic system, and the visual network. 
This high similarity between the two networks is in line with the emerging theory of functional connectome fingerprinting~\citep{finn2015functional}.
Leading consistent areas include the anterior cingulate cortex (ACC), the rostral middle frontal cortex, the medial orbitofrontal cortex, the temporal pole, and the ventromedial prefrontal cortex with a positive loading and the inferior parietal lobule, the lateral occipital cortex, the cuneus, the lingual gyrus, and the fusiform gyrus with a negative loading.
Regions in the prefrontal cortex and frontoparietal network are well known to be involved in working memory functions~\citep{funahashi2006prefrontal,barch2013function}. Hypothetically, functional connectivity between these regions is modulated by working memory loads during tasks and potentially a third brain region at rest. A recent study found that the ACC showed different modulatory interactions with regions in the network both in the resting-state and working memory tasks~\citep{di2020anterior}.
The cuneus, the lingual gyrus, and the fusiform gyrus (also known as the lateral occipitotemporal gyrus) are in the visual network involved in visual processing, visual memory, and face recognition. Coactivation in these regions has been observed in the processing of visual information and working memory tasks~\citep{bogousslavsky1987lingual,mccarthy1999electrophysiological,palejwala2021anatomy,sellal2022anatomical}.
% C3
C3 is a component with a low similarity between the two networks ($\text{similarity}=-0.095$). The rs-network consists of the DMN, including the middle temporal gyrus, and the superior frontal gyrus (Left), and the caudal middle frontal gyrus, and subcortical regions, including the parahippocampal (Left) and the caudate. The tb-network mainly consists of the visual network, the somato-motor network, the frontal-parietal network, and the limbic system. Regions include the lateral occipital cortex, the fusiform, the precentral and postcentral gyrus, the superior temporal gyrus, the posterior insula, the temporal gyrus and pole, the inferior temporal gyrus, the medial orbitofrontal cortex, and the rostral middle frontal cortex.
Primarily consisting of the visual and motor networks, the tb-network of C3 is a sensory binding and motor control component potentially activated to the distraction period between task blocks~\citep{dorfel2014common,li2021differential}.
In summary, the proposed approach identifies three orthogonal associated resting-state and task-state brain networks, one related to global signaling (C1), one related to the working memory task (C2), and one unrelated to the task (C3). 
%-------------------------------

%-------------------------------
% 221018/HCP-A/FACENAME-REST1AP
\begin{table}
    \caption{\label{tabel:hcp}Estimated model coefficients, $95\%$ confidence intervals from $500$ bootstrap samples, and the similarity between $\hat{\bgamma}$ and $\hat{\btheta}$ (denoted as $\langle\hat{\bgamma},\hat{\btheta}\rangle$) of the three identified components (C1, C2, and C3) using the proposed CoCReg approach in the HCP Aging study.}
    \begin{center}
        \resizebox{\textwidth}{!}{
        \begin{tabular}{l r r c r r c r r r}
            \hline
            & \multicolumn{2}{c}{rs-fMRI} && \multicolumn{2}{c}{Male$-$Female} && \multicolumn{2}{c}{Age} \\
            \cline{2-3}\cline{5-6}\cline{8-9}
            & \multicolumn{1}{c}{Estimate} & \multicolumn{1}{c}{$95\%$ CI} && \multicolumn{1}{c}{Estimate} & \multicolumn{1}{c}{$95\%$ CI} && \multicolumn{1}{c}{Estimate} & \multicolumn{1}{c}{$95\%$ CI} & \multicolumn{1}{c}{\multirow{-2}{*}{$\langle \hat{\bgamma},\hat{\btheta}\rangle$}} \\
            \hline
            C1 & $0.669$ & $(\textcolor{white}{-}0.572, \textcolor{white}{-}0.028)$ && $0.028$ & $(-0.009, \textcolor{white}{-}0.065)$ && $0.002$ & $(-0.017, \textcolor{white}{-}0.020)$ & $0.490$ \\
            C2 & $0.217$ & $(\textcolor{white}{-}0.142, \textcolor{white}{-}0.290)$ && $-0.067$ & $(-0.118, -0.015)$ && $0.043$ & $(\textcolor{white}{-}0.018, \textcolor{white}{-}0.069)$ & $0.663$ \\
            C3 & $-0.301$ & $(-0.407, -0.195)$ && $0.002$ & $(-0.052, \textcolor{white}{-}0.056)$ && $-0.053$ & $(-0.080, -0.027)$ & $-0.095$ \\
            \hline
        \end{tabular}
        }
    \end{center}
\end{table}

\begin{figure}
    \begin{center}
        \subfloat[C1-tb brain map ($\bgamma$)]{\includegraphics[width=0.25\textwidth]{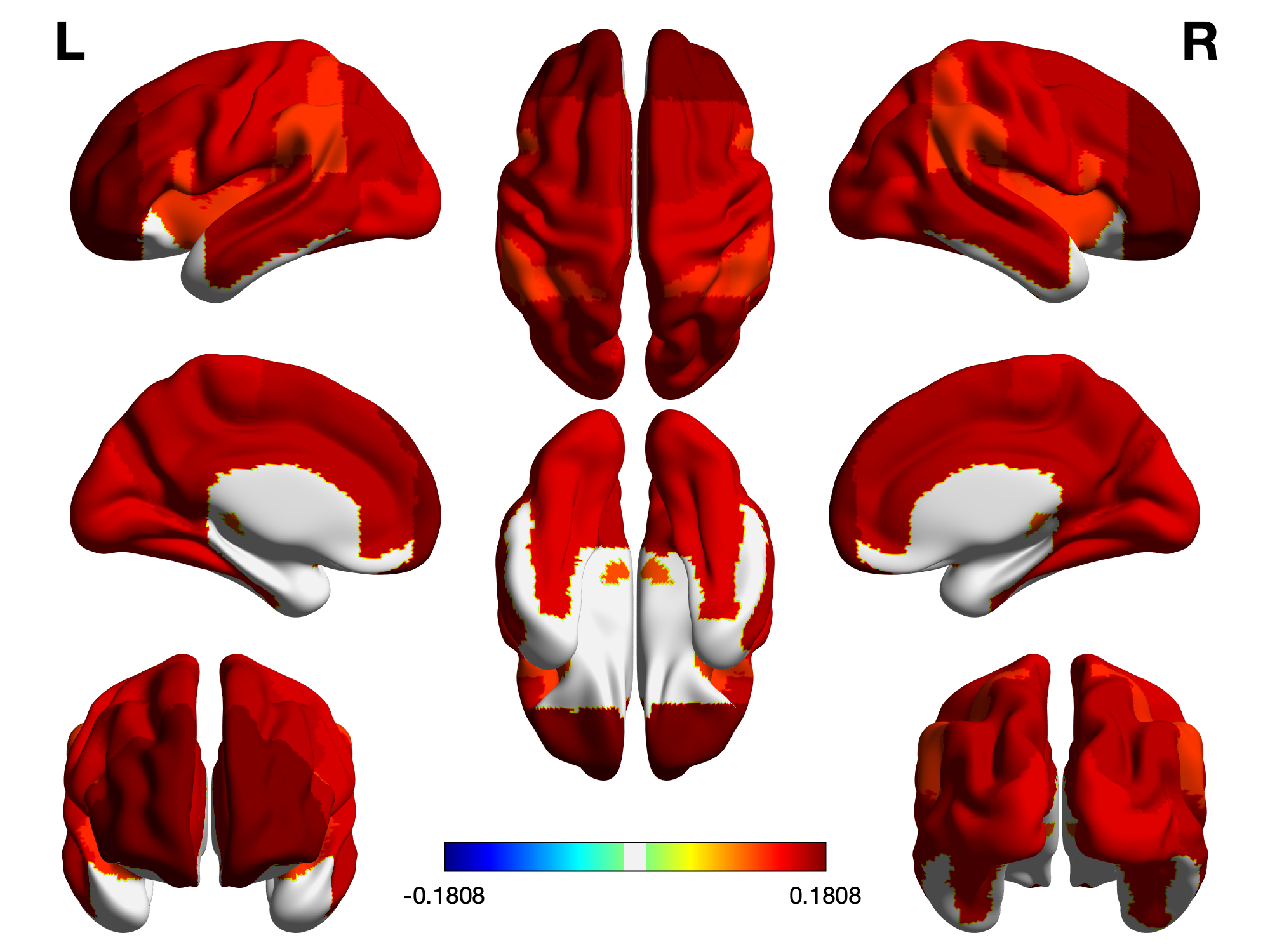}}
        \subfloat[C1-tb river plot ($\bgamma$)]{\includegraphics[width=0.25\textwidth]{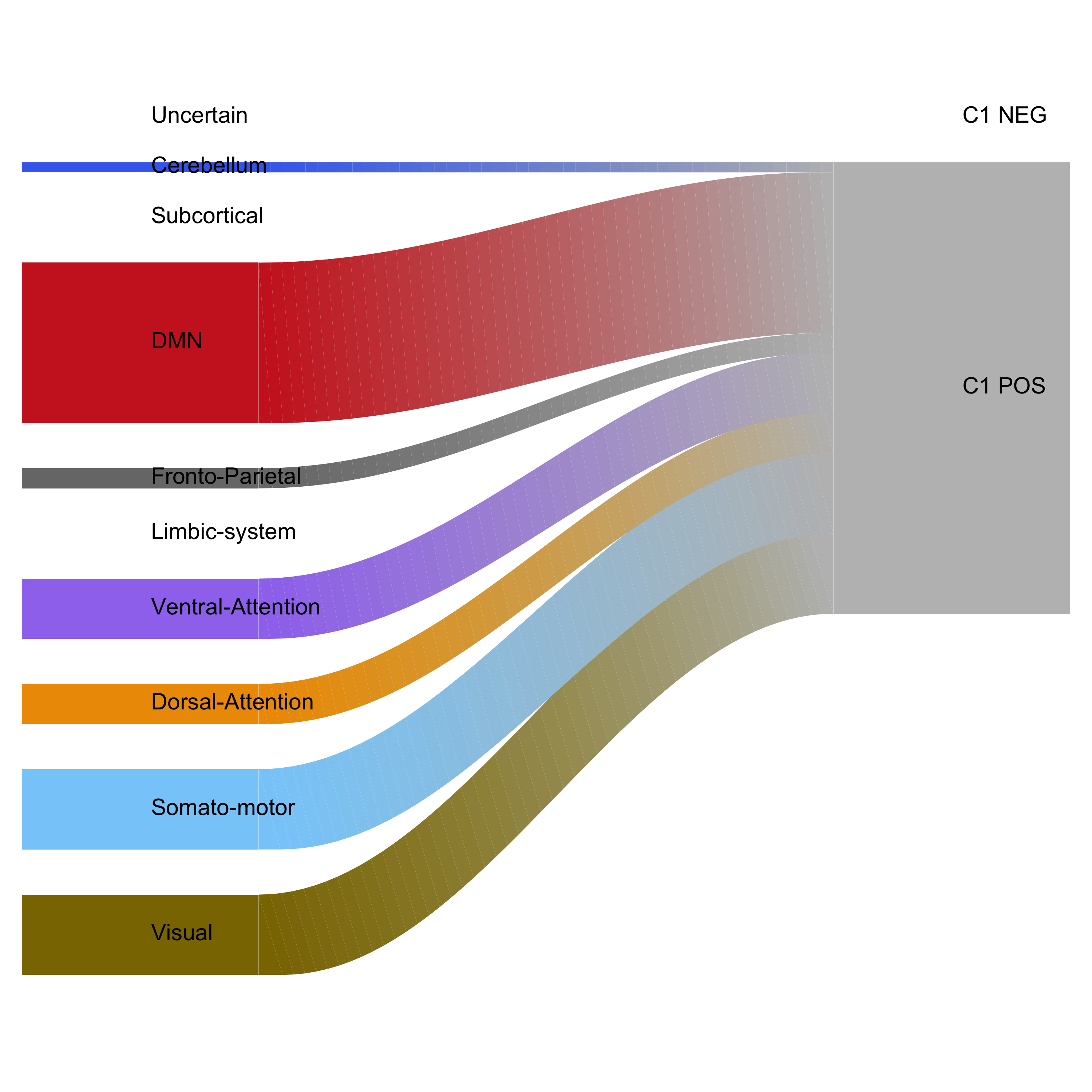}}
        \subfloat[C1-rs brain map ($\btheta$)]{\includegraphics[width=0.25\textwidth]{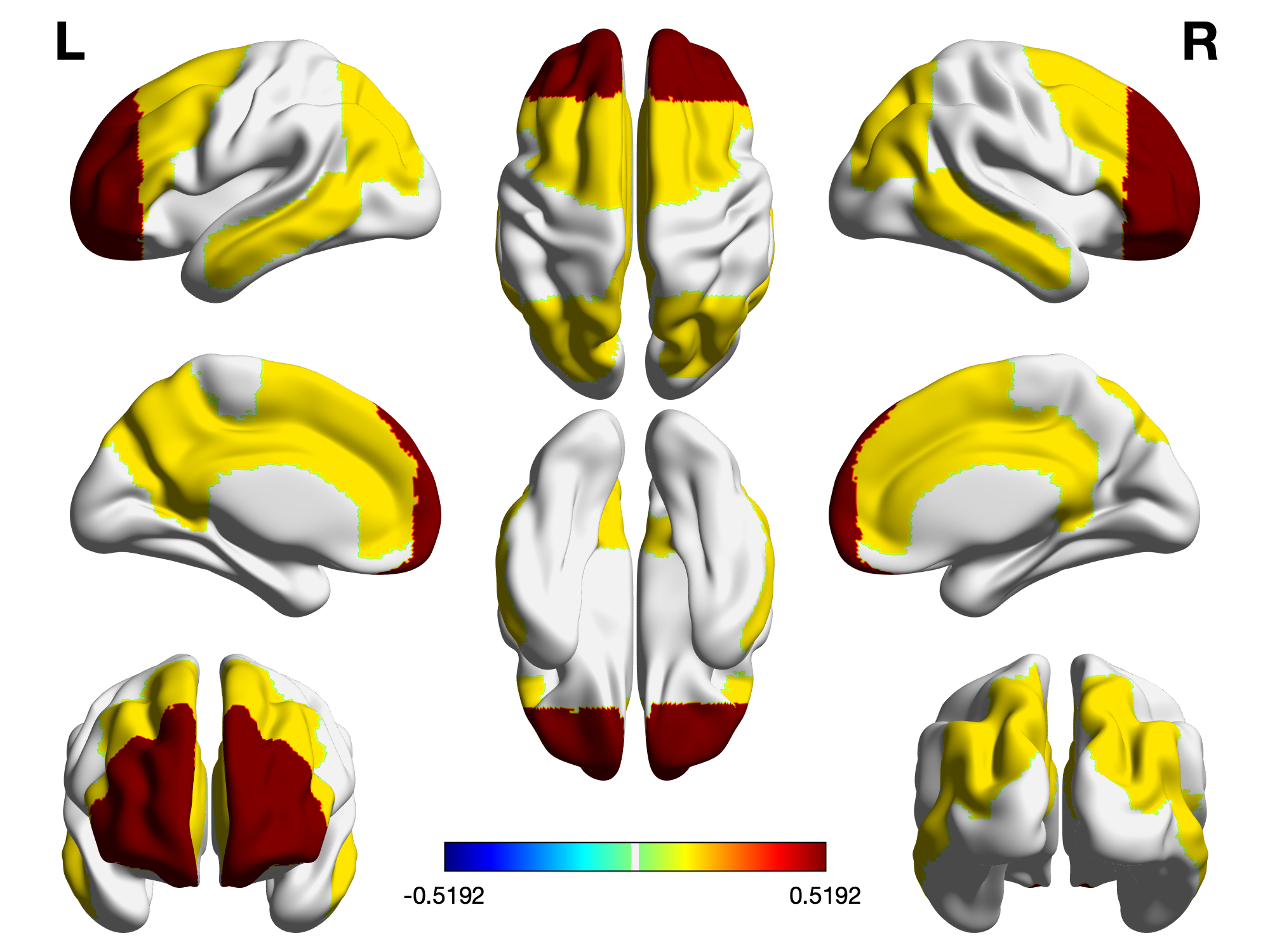}}
        \subfloat[C1-rs river plot ($\btheta$)]{\includegraphics[width=0.25\textwidth]{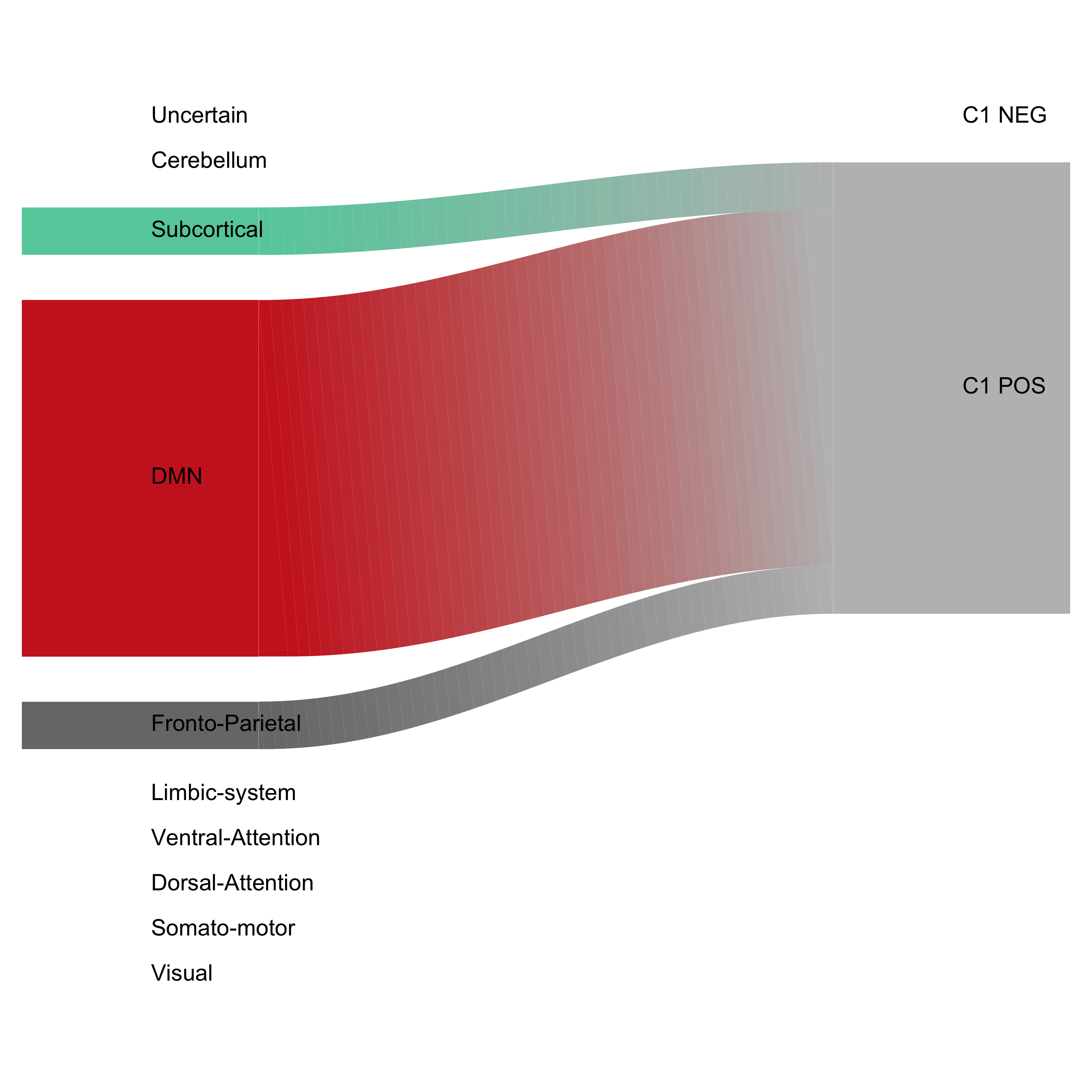}}

        \subfloat[C2-tb brain map ($\bgamma$)]{\includegraphics[width=0.25\textwidth]{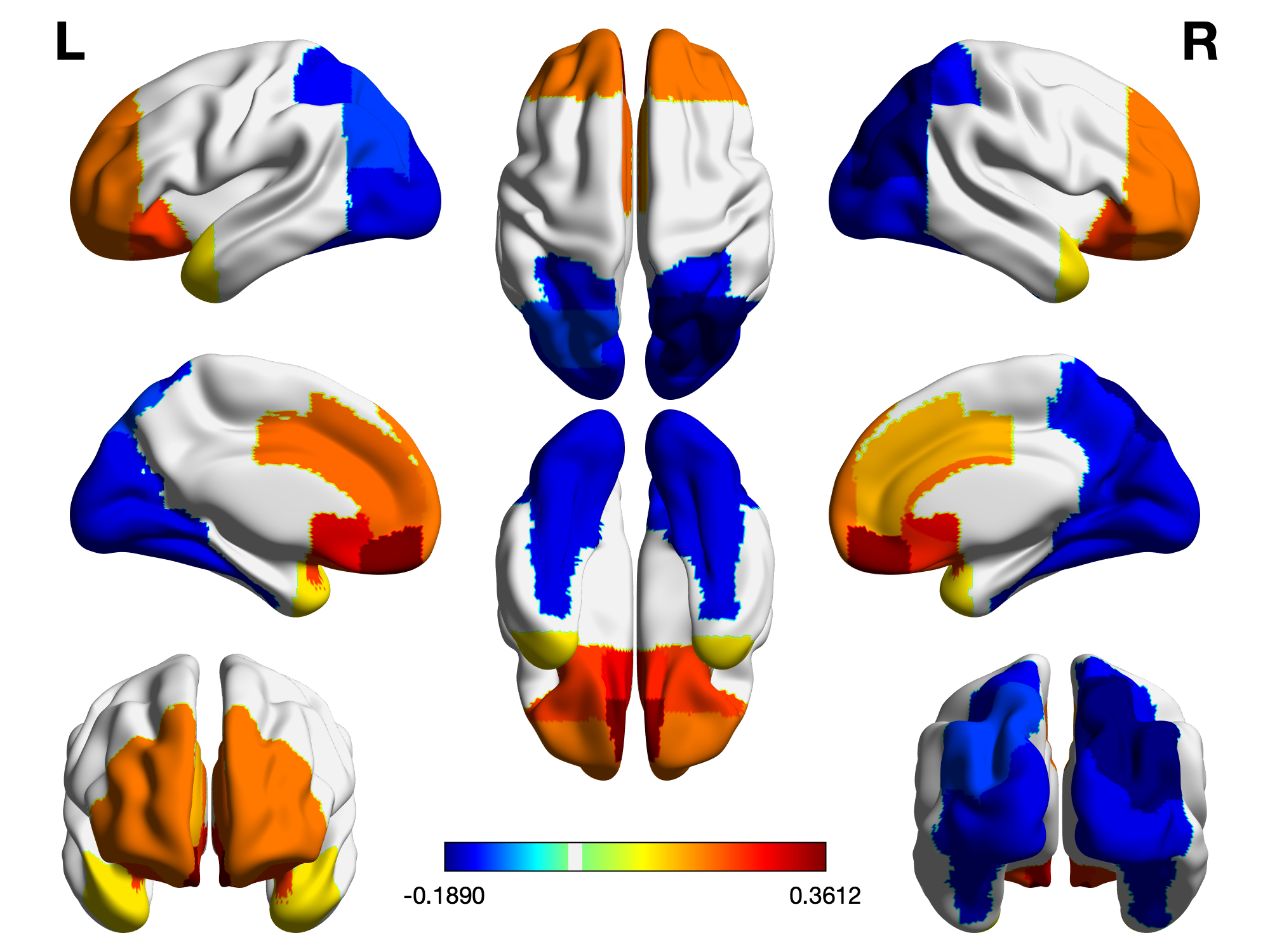}}
        \subfloat[C2-tb river plot ($\bgamma$)]{\includegraphics[width=0.25\textwidth]{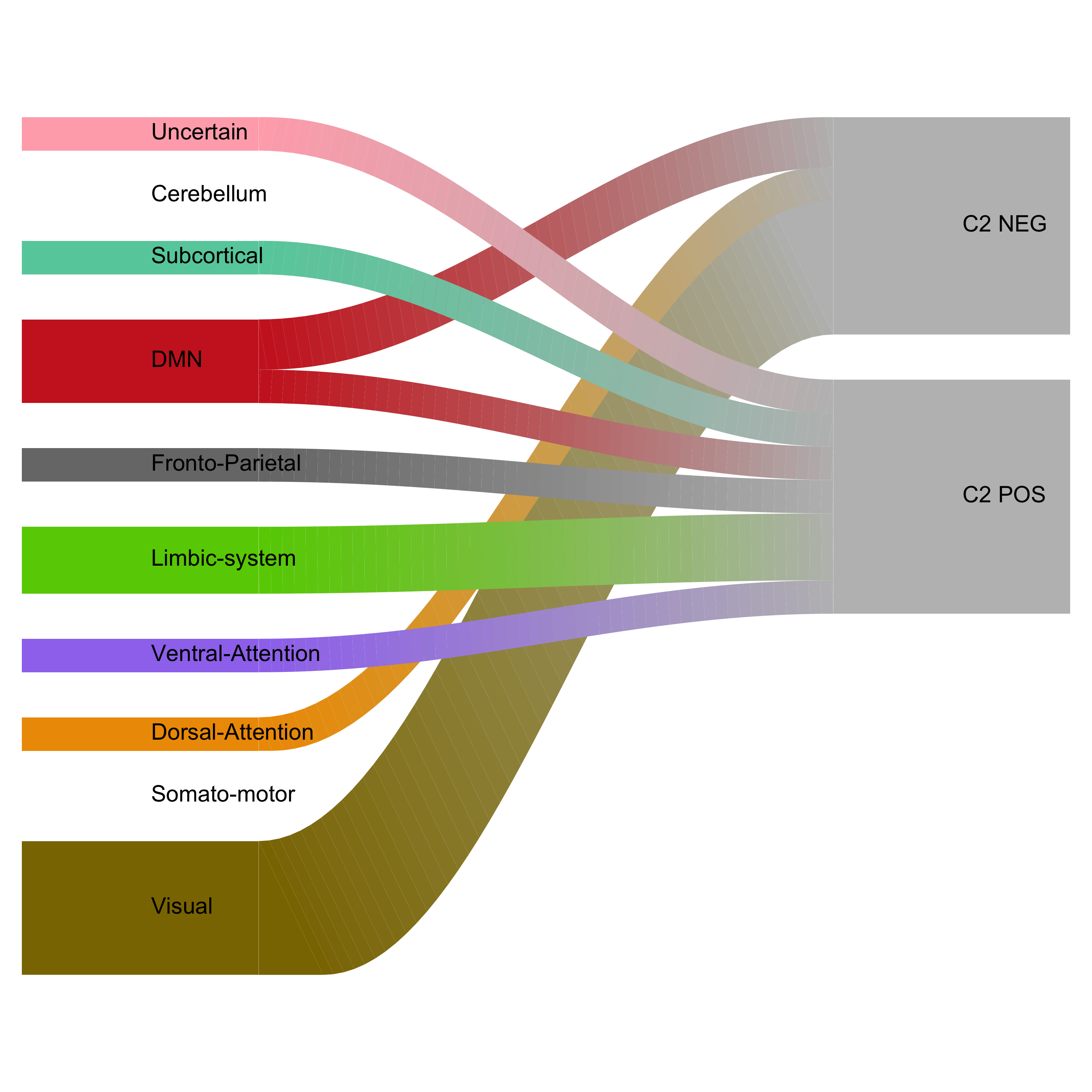}}
        \subfloat[C2-rs brain map ($\btheta$)]{\includegraphics[width=0.25\textwidth]{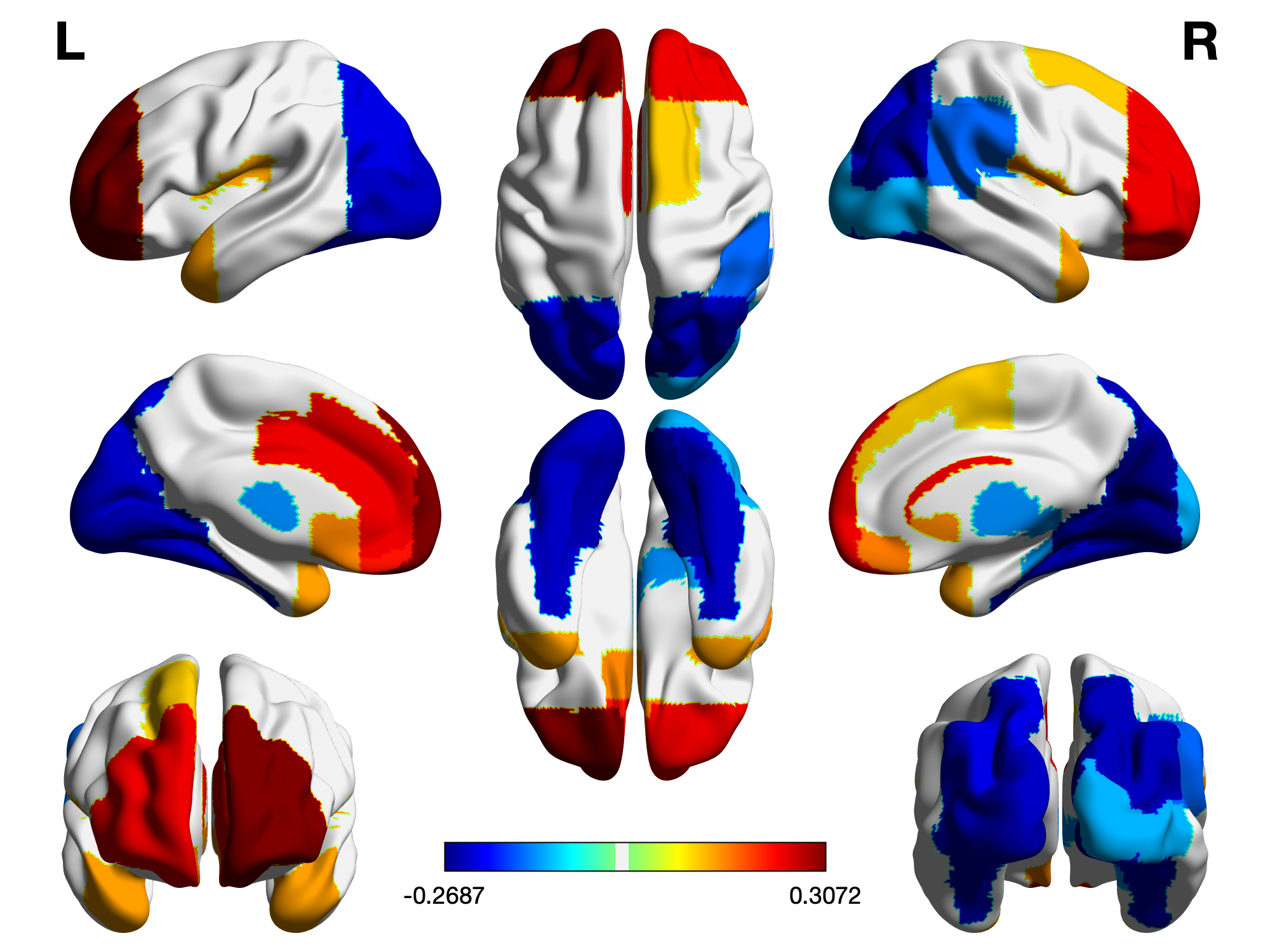}}
        \subfloat[C2-rs river plot ($\btheta$)]{\includegraphics[width=0.25\textwidth]{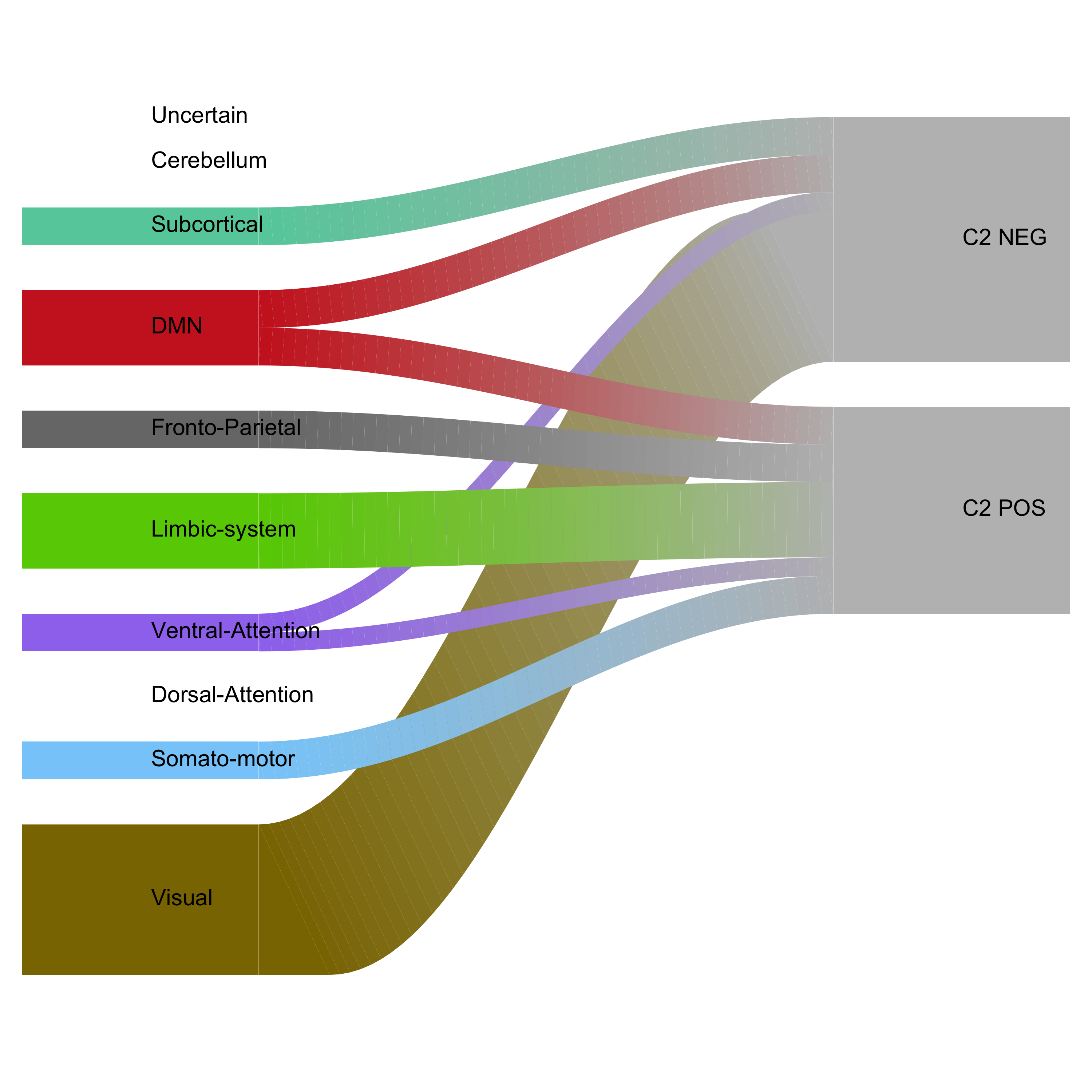}}

        \subfloat[C3-tb brain map ($\bgamma$)]{\includegraphics[width=0.25\textwidth]{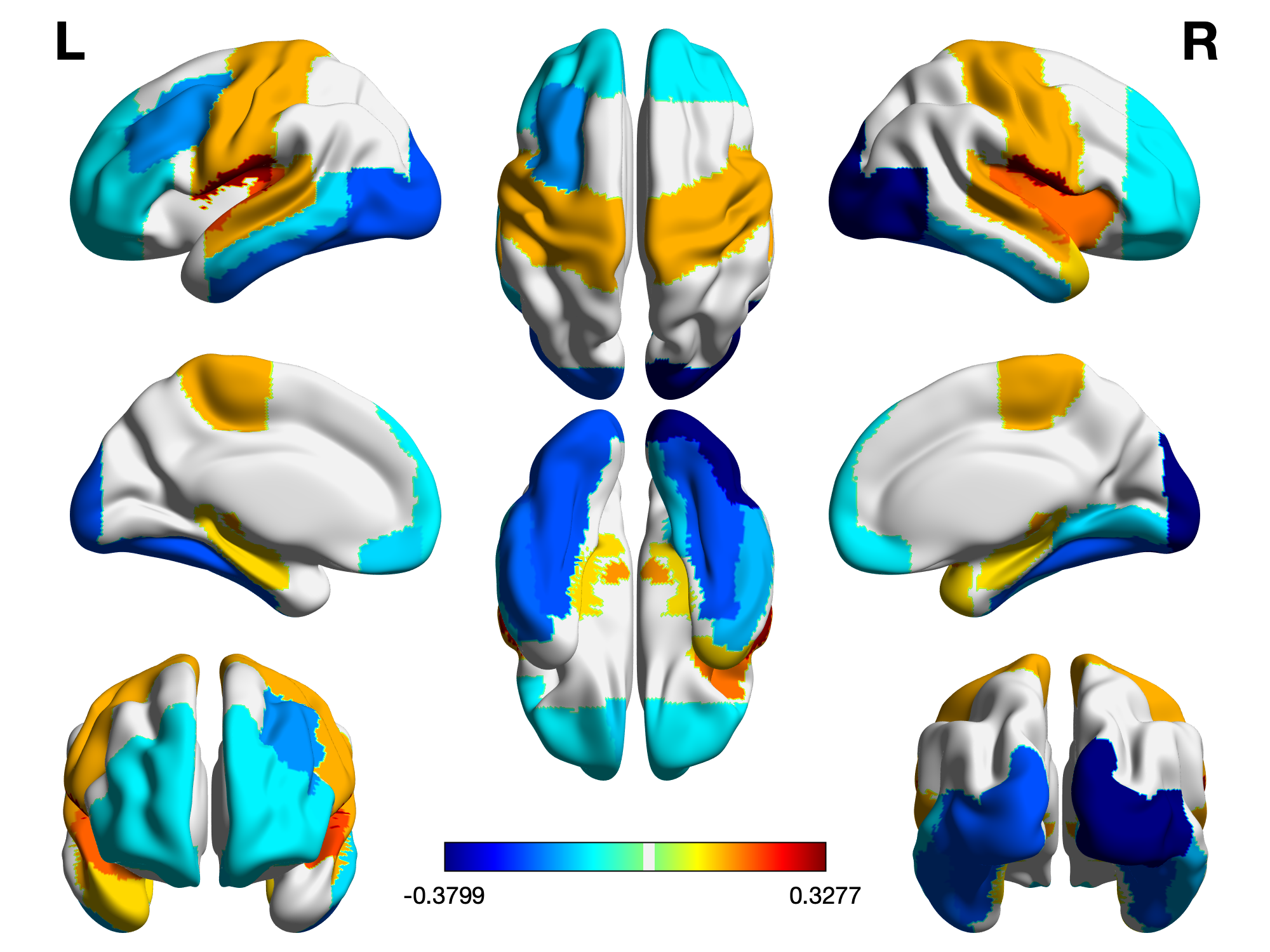}}
        \subfloat[C3-tb river plot ($\bgamma$)]{\includegraphics[width=0.25\textwidth]{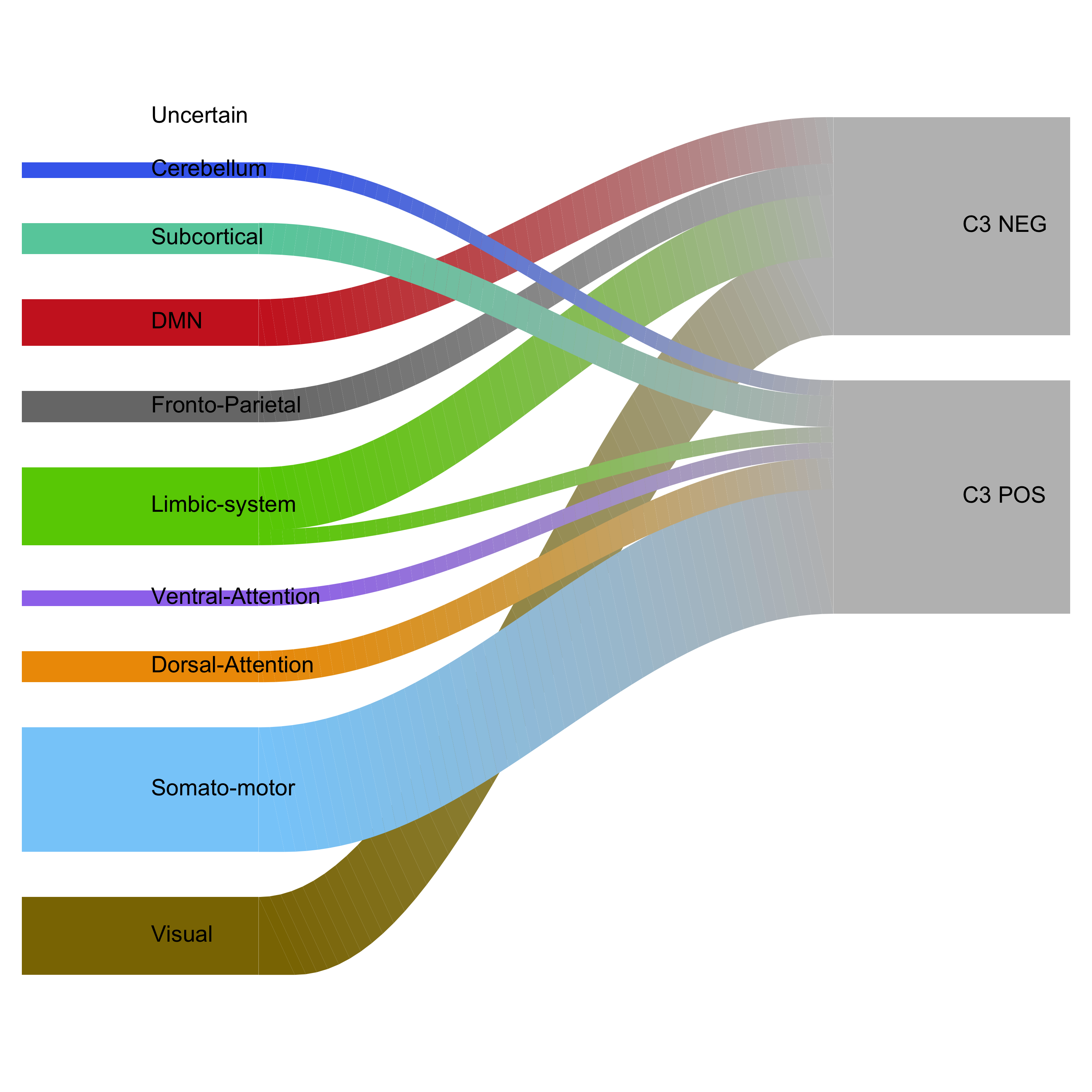}}
        \subfloat[C3-rs brain map ($\btheta$)]{\includegraphics[width=0.25\textwidth]{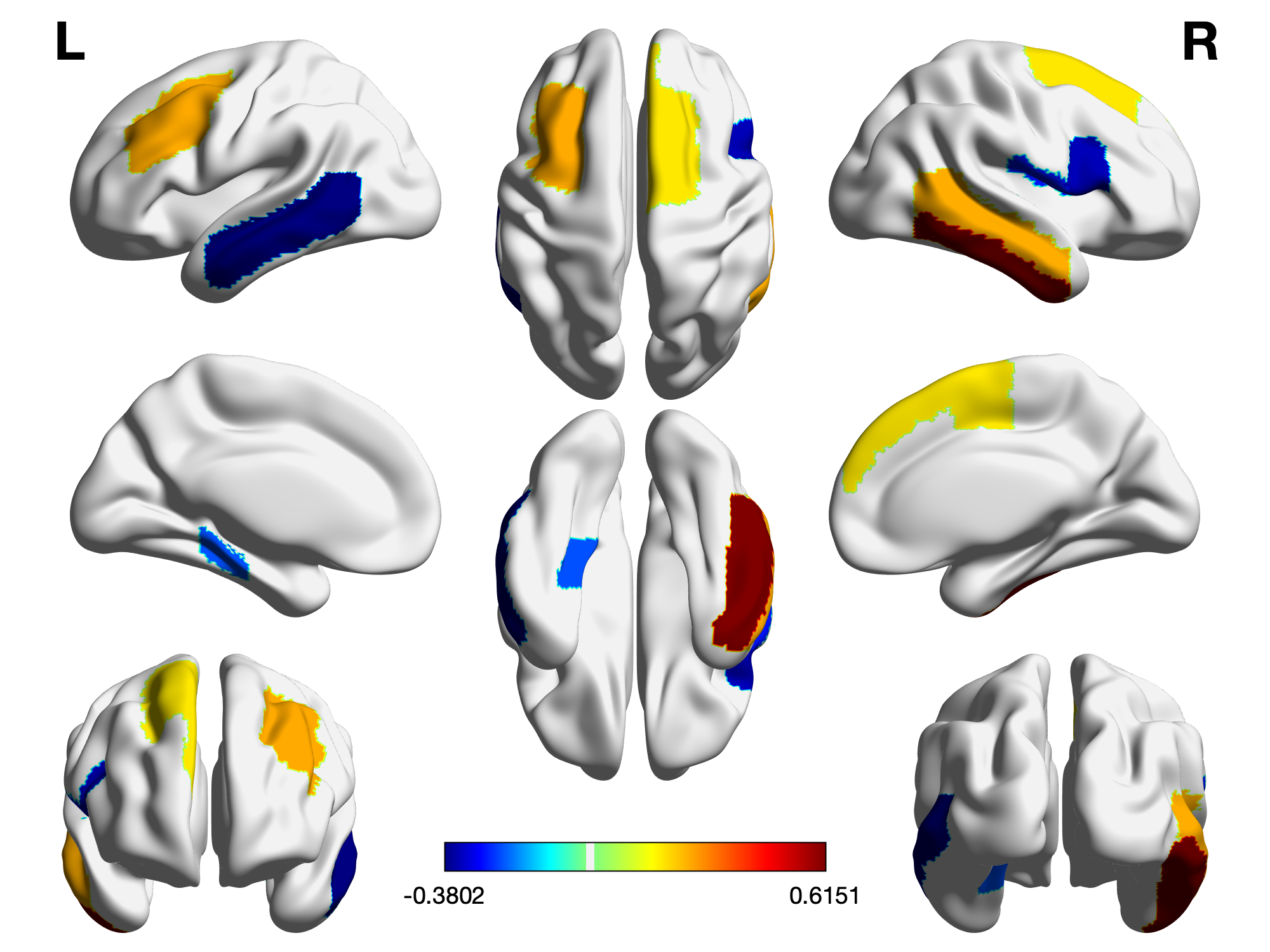}}
        \subfloat[C3-rs river plot ($\btheta$)]{\includegraphics[width=0.25\textwidth]{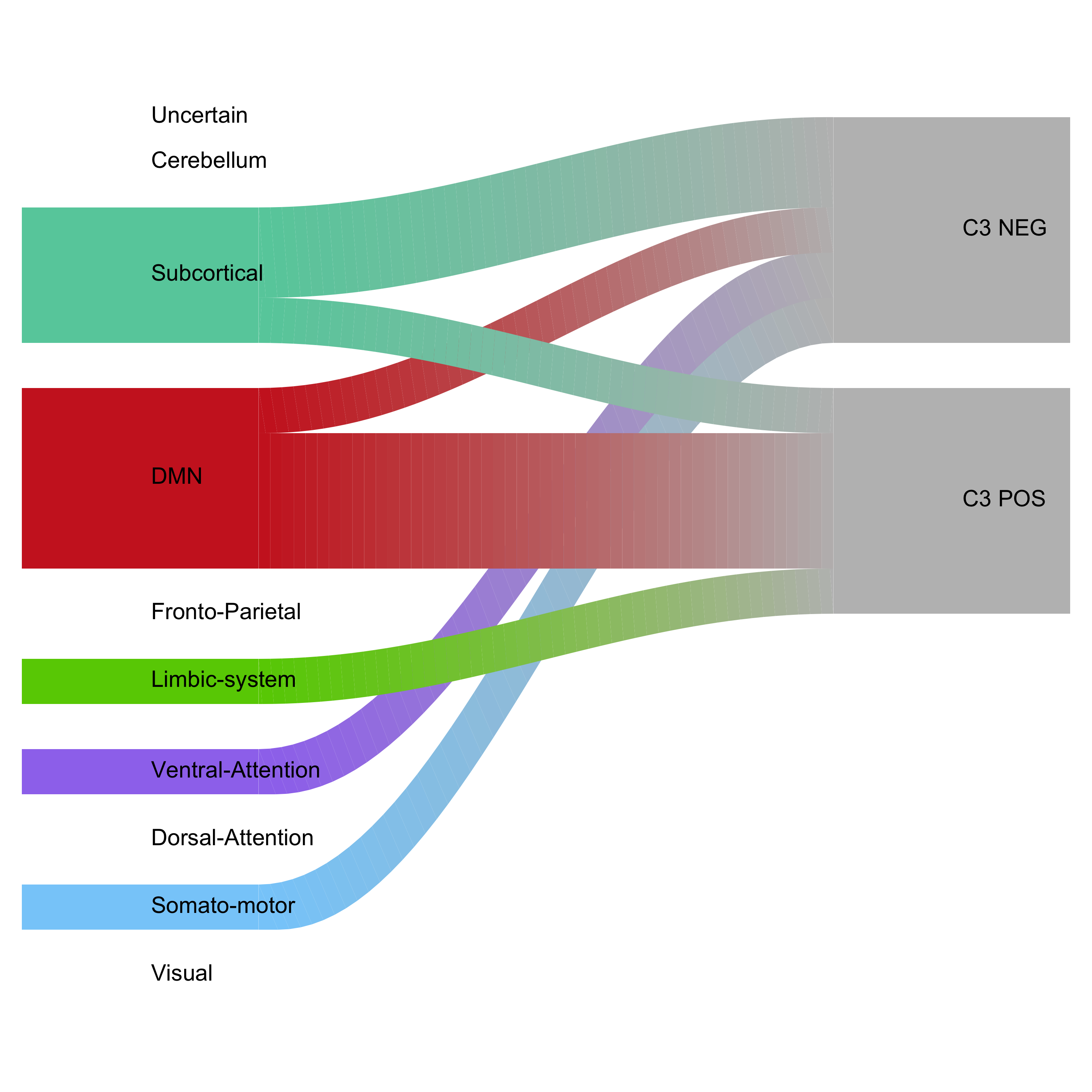}}
    \end{center}
    \caption{\label{fig:hcp}Regions with a nonzero loading in a brain map and the river plot of module configuration of the three identified components (C1, C2, and C3), using the proposed CoCReg approach in the HCP Aging study. tb: results of task-based fMRI; rs: results of resting-state fMRI.}
\end{figure}
%-------------------------------

%%%%%%%%%%%%%%%%%%%%%%%%%%%%%%%%%%%%%%%%%%%%%%%%%%%%%%%%%%

%%%%%%%%%%%%%%%%%%%%%%%%%%%%%%%%%%%%%%%%%%%%%%%%%%%%%%%%%%
% Discussion
%%%%%%%%%%%%%%%%%%%%%%%%%%%%%%%%%%%%%%%%%%%%%%%%%%%%%%%%%%
\section{Discussion}
\label{sec:discussion}

In this study, a Covariance-on-Covariance Regression (CoCReg) model is introduced. It is assumed that there exists a pair of linear projections on the outcome covariance matrices and the predictor covariance matrices, such that in the projection spaces, a log-linear model is satisfied to associate the variances. An ordinary least squares type of estimator is introduced for simultaneous projection identification and model coefficient estimation. Under regularity conditions, the proposed estimator is asymptotically consistent. Simulation studies demonstrate the superior performance of the proposed approach over a modified existing method. Applying to data collected in the HCP Aging study, the proposed approach identifies three pairs of networks, where functional connectivity within the resting-state network can predict functional connectivity within the corresponding task-state network. The three networks correspond to a global signal network, a task-related network, and a task-unrelated network. The findings are consistent with existing knowledge about brain function and activation at rest and during a memory task.

The asymptotic consistency of the proposed estimator is achieved under the assumption of complete common diagonalization for the outcome covariance matrices and the predictor covariance matrices. In \citet{zhao2021covariate}, a complete common diagonalization was also assumed. Via simulation studies, it demonstrated that this assumption can be relaxed to partial common diagonalization. In Section~\ref{sec:sim}, a partial common diagonalization scenario is also considered and the numerical result suggests that the proposed approach is robust to this relaxation. Thus, a theoretical study of the asymptotic consistency under this relaxation is one future direction.
The proposed framework considers the low-dimensional scenario, where the dimensions of the outcome covariance matrix and the predictor covariance matrix are both lower than the number of observations acquired from each subject. Asymptotic properties are investigated under this scenario. One extension is to consider the case of high-dimensional data. In \citet{zhao2021principal}, a shrinkage estimator of the outcome covariance matrices was introduced, where the shrinkage parameters were assumed to be shared across subjects to yield optimal performance. Shrinkage estimators of the outcome covariance matrices and the predictor covariance matrices can be then considered analogously and we leave it as our future research.
A bootstrap procedure is introduced for inference on the model coefficients, not on the linear projections. As for each bootstrap sample, the identified projections and the order of identifying these projections may differ, performing inference on these projection vectors requires a matching procedure. The inference will then highly depend on the performance of matching and the metric used for matching. Thus, a more thorough theoretical and numerical investigation is necessary. As this is beyond the scope of the current manuscript, it will be studied in the future.
%%%%%%%%%%%%%%%%%%%%%%%%%%%%%%%%%%%%%%%%%%%%%%%%%%%%%%%%%%

%%%%%%%%%%%%%%%%%%%%%%%%%%%%%%%%%%%%%%%%%%%%%%%%%%%%%%%%%%
% \section*{Acknowledgments}

%----------------------
%----------------------
%%%%%%%%%%%%%%%%%%%%%%%%%%%%%%%%%%%%%%%%%%%%%%%%%%%%%%%%%%

%%%%%%%%%%%%%%%%%%%%%%%%%%%%%%%%%%%%%%%%%%%%%%%%%%%%%%%%%%
% Appendix
%%%%%%%%%%%%%%%%%%%%%%%%%%%%%%%%%%%%%%%%%%%%%%%%%%%%%%%%%%
% \clearpage

\appendix
\counterwithin{figure}{section}
\counterwithin{table}{section}
\counterwithin{equation}{section}
\counterwithin{lemma}{section}
\counterwithin{theorem}{section}

%%%%%%%%%%%%%%%%%%%%%%%%%%%%%%%%%%%%%%%%%%%%%%%%%%%%%%%%%%
% Theory and proof
%%%%%%%%%%%%%%%%%%%%%%%%%%%%%%%%%%%%%%%%%%%%%%%%%%%%%%%%%%
\section{Theory and proof}

%========================================================%
\subsection{Details of Algorithm~\ref{alg:obj_solve}}
\label{appendix:sub:obj_solve}

This section provides the details of solving optimization problem~\eqref{eq:obj}. The loss function, $\ell$, is bi-convex over parameters $(\bgamma,\btheta,\alpha,\bbeta)$. Thus, one can solve for the solutions by coordinate descent.

For $\alpha$,
\[
    \frac{\partial\ell}{\partial\alpha}=\frac{1}{n}\sum_{i=1}^{n}2\left\{\log(\bgamma^\top\hat{\Sigma}_{i}\gamma)-\alpha\log(\btheta^\top\hat{\bDelta}_{i}\btheta)-\bbeta^\top\bw_{i}\right\}\left\{-\log(\btheta^\top\hat{\bDelta}_{i}\btheta)\right\}=0,
\]
\[
    \Rightarrow \quad \alpha=\left\{\frac{1}{n}\sum_{i=1}^{n}\log^{2}(\btheta^\top\hat{\bDelta}_{i}\btheta)\right\}^{-1}\left[\frac{1}{n}\sum_{i=1}^{n}\left\{\log(\bgamma^\top\hat{\bSigma}_{i}\bgamma)\log(\btheta^\top\hat{\bDelta}_{i}\btheta)-(\bbeta^\top\bw_{i})\log(\btheta^\top\hat{\Delta}_{i}\btheta)\right\}\right].
\]

For $\bbeta$,
\[
    \frac{\partial\ell}{\partial\bbeta}=\frac{1}{n}\sum_{i=1}^{n}2\left\{\log(\bgamma^\top\hat{\bSigma}_{i}\bgamma)-\alpha\log(\btheta^\top\hat{\bDelta}_{i}\btheta)-\bbeta^\top\bw_{i}\right\}(-\bw_{i}^\top)=\boldsymbol{\mathrm{0}},
\]
\[
    \Rightarrow \quad \bbeta=\left(\frac{1}{n}\sum_{i=1}^{n}\bw_{i}\bw_{i}^\top\right)^{-1}\left[\frac{1}{n}\sum_{i=1}^{n}\left\{\log(\bgamma^\top\hat{\bSigma}_{i}\bgamma)-\alpha\log(\btheta^\top\hat{\bDelta}_{i}\btheta)\right\}\bw_{i}\right].
\]

For $\bgamma$, let $U_{i}=\alpha\log(\btheta^\top\hat{\bDelta}_{i}\btheta)+\bbeta^\top\bw_{i}$. Under the constraint, the Lagrangian form is
\[
    \mathcal{L}(\bgamma)=\frac{1}{n}\sum_{i=1}^{n}\left\{\log(\bgamma^\top\hat{\bSigma}_{i}\bgamma)-U_{i}\right\}^{2}-\lambda_{1}(\bgamma^\top\bH_{y}\bgamma-1)
\]
\begin{eqnarray*}
    \frac{\partial\mathcal{L}}{\partial\bgamma} &=& \frac{2}{n}\sum_{i=1}^{n}\left\{\log(\bgamma^\top\hat{\bSigma}_{i}\bgamma)-U_{i}\right\}\frac{2\hat{\bSigma}_{i}\bgamma}{\bgamma^\top\hat{\bSigma}_{i}\bgamma}-2\lambda_{1}\bH_{y}\bgamma=\boldsymbol{\mathrm{0}}, \\
    \frac{\partial\mathcal{L}}{\partial\lambda_{1}} &=& \bgamma^\top\bH_{y}\bgamma-1=0.
\end{eqnarray*}
Plugging in $\bgamma$ from previous step $h$ into $\bgamma^\top\hat{\bSigma}_{i}\bgamma$ denoted as $\xi_{i}=\bgamma^{(h)\top}\hat{\bSigma}_{i}\bgamma^{(h)}$,
\[
    \Rightarrow \quad \left(\frac{2}{n}\sum_{i=1}^{n}\frac{\log\xi_{i}-U_{i}}{\xi_{i}}\hat{\bSigma}_{i}\right)\bgamma-\lambda_{1}\bH_{y}\bgamma\triangleq\bA_{1}\bgamma-\lambda_{1}\bH_{y}\bgamma=\boldsymbol{\mathrm{0}},
\]
where
\[
    \bA_{1}=\frac{2}{n}\sum_{i=1}^{n}\frac{\log\xi_{i}-U_{i}}{\xi_{i}}\hat{\bSigma}_{i}.
\]
The solution $(\bgamma,\lambda_{1})$ is the eigenvector and eigenvalue of $\bA_{1}$ with respect to $\bH_{y}$~\cite[see details in the supplementary mateirals of ][]{zhao2021covariate}.

For $\btheta$, analogous to the solution to $\bgamma$, let $V_{i}=\log(\bgamma^\top\hat{\bSigma}_{i}\bgamma)-\bbeta^\top\bw_{i}$. Under the constraint, the Lagrangian form is
\[
    \mathcal{L}(\btheta)=\frac{1}{n}\sum_{i=1}^{n}\left\{\alpha\log(\btheta^\top\hat{\bDelta}_{i}\btheta)-V_{i}\right\}^{2}-\lambda_{2}(\btheta^\top\bH_{x}\btheta-1).
\]
\begin{eqnarray*}
    \frac{\partial\mathcal{L}}{\partial\btheta} &=& \frac{2}{n}\sum_{i=1}^{n}\left\{\alpha\log(\btheta^\top\hat{\bDelta}_{i}\btheta)-V_{i}\right\}\frac{2\alpha\hat{\bDelta}_{i}\btheta}{\btheta^\top\hat{\bDelta}_{i}\btheta}-2\lambda_{2}\bH_{x}\btheta=\boldsymbol{\mathrm{0}} \\
    \frac{\partial\mathcal{L}}{\partial\lambda_{2}} &=& \btheta^\top\bH_{x}\btheta-1=0
 \end{eqnarray*}
 Plugging in $\btheta$ from previous step $h$ into $\btheta^\top\hat{\bDelta}_{i}\btheta$ denoted as $\zeta_{i}=\btheta^{(h)\top}\hat{\bDelta}_{i}\btheta^{(h)}$,
 \[
    \Rightarrow\quad \left(\frac{2\alpha}{n}\sum_{i=1}^{n}\frac{\alpha\log\zeta_{i}-V_{i}}{\zeta_{i}}\hat{\bDelta}_{i}\right)\btheta-\lambda_{2}\bH_{x}\btheta\triangleq\bA_{2}\btheta-\lambda_{2}\bH_{x}\btheta=\boldsymbol{\mathrm{0}},
 \] 
 where
 \[
    \bA_{2}=\frac{2\alpha}{n}\sum_{i=1}^{n}\frac{\alpha\log\zeta_{i}-V_{i}}{\zeta_{i}}\hat{\bDelta}_{i}.
 \]
 The solution $(\btheta,\lambda_{2})$ is the eigenvector and eigenvalue of $\bA_{2}$ with respect to $\bH_{x}$.
%========================================================%

%========================================================%
\subsection{Asymptotic properties}
\label{appendix:sub:asmp}

Theorem~\ref{thm:asmp} presents the asymptotic distribution of the proposed OLS estimator of $(\alpha,\bbeta)$ when the projections are known.

\begin{theorem}\label{thm:asmp}
    Assume Assumptions A1--A4 in Section~\ref{sub:asmp} hold. For given $(\bgamma,\btheta)$, assume
    \begin{equation}
        \frac{1}{n}\sum_{i=1}^{n}\log^{2}\left\{\btheta^\top\left(\frac{1}{u_{i}}\sum_{s=1}^{u_{i}}\bx_{is}
        \bx_{is}^\top\right)\btheta\right\}\rightarrow G_{x}\in\mathbb{R}, \quad \text{as } n\rightarrow\infty, ~u\rightarrow\infty,
    \end{equation}
    \begin{equation}
        \frac{1}{n}\sum_{i=1}^{n}\bw_{i}\bw_{i}^\top\rightarrow \bQ_{w}\in\mathbb{R}^{r\times r}, \quad \text{as } n\rightarrow\infty,
    \end{equation}
    \begin{equation}
        \frac{1}{n}\sum_{i=1}^{n}\log\left\{\btheta^\top\left(\frac{1}{u_{i}}\sum_{s=1}^{u_{i}}\bx_{is}\bx_{is}^\top\right)\btheta\right\}\bw_{i}\rightarrow \bH_{xw}\in\mathbb{R}^{r\times 1}, \quad \text{as } n\rightarrow\infty, ~u\rightarrow\infty.
    \end{equation}
    Let $M_{n}=\sum_{i=1}^{n}v_{i}$, as $n,u,v\rightarrow\infty$,
    \begin{equation}\label{eq:asmp_coef}
        \sqrt{M_{n}}\left(\begin{pmatrix}
            \hat{\alpha} \\
            \hat{\bbeta}
        \end{pmatrix}-\begin{pmatrix}
            \alpha \\
            \bbeta
        \end{pmatrix}\right)\overset{\mathcal{D}}{\longrightarrow}\mathcal{N}\left(\boldsymbol{\mathrm{0}},\begin{pmatrix}
            G_{x} & \bH_{xw}^\top \\
            \bH_{xw} & \bQ_{w}
        \end{pmatrix}^{-1}\right).
    \end{equation}
\end{theorem}

\begin{proof}
    Let
    \[
        \ell_{i}=\left\{\log(\bgamma^\top\hat{\bSigma}_{i}\bgamma)-\alpha\log(\btheta^\top\hat{\bDelta}_{i}\btheta)-\bbeta^\top\bw_{i}\right\}^{2}.
    \]
    \begin{eqnarray*}
        \frac{\partial^{2}\ell_{i}}{\partial\alpha^{2}} &=& \log^{2}(\btheta^\top\hat{\bDelta}_{i}\btheta), \\
        \frac{\partial^{2}\ell_{i}}{\partial\bbeta\partial\bbeta^\top} &=& \bw_{i}\bw_{i}^\top, \\
        \frac{\partial^{2}\ell_{i}}{\partial\alpha\partial\bbeta^\top} &=& \log(\btheta^\top\hat{\bDelta}_{i}\btheta)\bw_{i}^\top.
    \end{eqnarray*}
    \[
        \frac{1}{n}\sum_{i=1}^{n}\log^{2}\left\{\btheta^\top\left(\frac{1}{u_{i}}\sum_{s=1}^{u_{i}}\bx_{is}
        \bx_{is}^\top\right)\btheta\right\}\rightarrow G_{x}\in\mathbb{R}, \quad \text{as } n\rightarrow\infty, u\rightarrow\infty,
    \]
    \[
        \frac{1}{n}\sum_{i=1}^{n}\bw_{i}\bw_{i}^\top\rightarrow \bQ\in\mathbb{R}^{r\times r}, \quad \text{as } n\rightarrow\infty,
    \]
    \[
        \frac{1}{n}\sum_{i=1}^{n}\log\left\{\btheta^\top\left(\frac{1}{u_{i}}\sum_{s=1}^{u_{i}}\bx_{is}\bx_{is}^\top\right)\btheta\right\}\bw_{i}\rightarrow \bH_{xw}\in\mathbb{R}^{r\times 1}, \quad \text{as } n\rightarrow\infty, ~u\rightarrow\infty.
    \]
    The asymptotic distribution in~\eqref{eq:asmp_coef} follows.
\end{proof}
%========================================================%

%========================================================%
\subsection{Proof of Proposition~\ref{prop:asmp}}
\label{appendix:sub:thm_asmp}

We first discuss the imposed assumptions. Assumption A1 assumes a low-dimensional scenario and the data dimensions are fixed. Assumption A2 regulates the data on higher-order moments. Usually, for resting-state data, normality can be assumed. Then, for $\bx_{is}$'s, this assumption is satisfied. For task-based data, the task-related activation is generally modeled as a convolution of the canonical hemodynamic response function (HRF) and the event onsite~\citep{lindquist2008statistical}. The fMRI signal can be considered as a linear superposition of the task-free random fluctuation and the task-related activation~\citep{cole2014intrinsic}. Thus, this higher-order moment constraint can be assumed to be satisfied. Assumption A3 assumes a common diagonalization for the covariance matrices, through the corresponding eigenvalue can be different. Assumption A4 assumes that the model is correctly specified. Under Assumptions A1--A3, the eigenvectors of $\bar{\bS}_{y}$ and $\bar{\bS}_{x}$ are consistent estimators of $\bPi$ and $\bUpsilon$, respectively.
The proposed estimators are $M$-estimators. Under regularity conditions (A1)--(A4), the consistency of the estimators follows.
%========================================================%

%========================================================%
%%%%%%%%%%%%%%%%%%%%%%%%%%%%%%%%%%%%%%%%%%%%%%%%%%%%%%%%%%

%%%%%%%%%%%%%%%%%%%%%%%%%%%%%%%%%%%%%%%%%%%%%%%%%%%%%%%%%%
% Simulation
%%%%%%%%%%%%%%%%%%%%%%%%%%%%%%%%%%%%%%%%%%%%%%%%%%%%%%%%%%
\section{Additional simulation results}
\label{appendix:sec:sim}

%========================================================%
\subsection{Non-Gaussian distributed data}
\label{appendix:sub:sim_nonGaussian}

In this section, the performance of the proposed approach is examined when the data distribution is non-Gaussian. Two multivariate distributions are considered: (1) multivariate $t$-distribution with degrees of freedom $\nu=3$ ($\text{skewness}=0$) and (2) matrix gamma distribution with shape parameter $\alpha=1$ ($\text{skewness}\neq 0$). For the multivariate $t$-distribution, the covariance matrices are generated following the same procedure as in Section~\ref{sec:sim}. For the matrix gamma distribution, the covariance matrices are first generated following the settings in Section~\ref{sec:sim}. Denote the covariance matrix, $\bDelta_{i}=(\delta_{ijk})$ and $\bSigma_{i}=(\sigma_{ilm})$. The shape parameter is set to one and the scale parameters are $\sqrt{\delta_{ijj}}$ for ($j=1,\dots,p$) and $\sqrt{\sigma_{ill}}$ (for $l=1,\dots,q$) and , such that the corresponding variances are $\delta_{ijj}$ and $\sigma_{ill}$, respectively. The correlation structure is set to be the corresponding correlation matrix of $\bDelta_{i}$ and $\bSigma_{i}$. For gamma distributed data, the expectation is nonzero. Thus, before applying the proposed approach, the data are centered to have mean zero. Table~\ref{appendix:table:sim_nonGaussian} presents the results with $p=10$ and $q=5$ under the sample size of $(n,u,v)=(100,100,100)$. From the table, for both multivariate $t$ and matrix gamma distributions, the proposed approach yields a good estimate of the parameters with a higher estimation bias compared to the results under the Gaussian distributions (Table~\ref{table:sim_est}). The proposed estimator is an OLS type estimator, where no distribution assumption is imposed. The simulation results demonstrate the robustness of the proposed estimator to non-Gaussian distributions.

% multivariate t: 220510/case5
% matrix gamma: 220510/case4
\begin{table}
    \caption{\label{appendix:table:sim_nonGaussian}Performance in identifying target components and estimating model coefficients for non-Gaussian data. Data dimension of $p=10$ and $q=5$ and sample size of $(n,u,v)=(100,100,100)$. SE: standard error; MSE: mean squared error.}
    \begin{center}
        \resizebox{\textwidth}{!}{
        \begin{tabular}{l l r r r r r c r r r}
            \hline
            & & & & \multicolumn{3}{c}{$\hat{\alpha}$} && \multicolumn{3}{c}{$\hat{\beta}_{1}$} \\
            \cline{5-7}\cline{9-11}
            \multicolumn{1}{c}{\multirow{-2}{*}{Distribution}} & & \multicolumn{1}{c}{\multirow{-2}{*}{$|\langle\hat{\bgamma},\bgamma\rangle|$ (SE)}} & \multicolumn{1}{c}{\multirow{-2}{*}{$|\langle\hat{\btheta},\btheta\rangle|$ (SE)}} & \multicolumn{1}{c}{Bias} & \multicolumn{1}{c}{SE} & \multicolumn{1}{c}{MSE} && \multicolumn{1}{c}{Bias} & \multicolumn{1}{c}{SE} & \multicolumn{1}{c}{MSE} \\
            \hline
            & C1 & $0.966$ ($0.085$) & $0.854$ ($0.137$) & $-0.172$ & $0.936$ & $0.897$ && $0.128$ & $0.200$ & $0.056$ \\
            \multirow{-2}{*}{Multivariate $t$ ($\nu=3$)} & C2 & $0.958$ ($0.075$) & $0.771$ ($0.150$) & $0.093$ & $1.459$ & $2.114$ && $-0.320$ & $0.306$ & $0.195$ \\
            \hline
            & C1 & $0.978$ ($0.072$) & $0.949$ ($0.072$) & $-0.277$ & $0.191$ & $0.113$ && $0.071$ & $0.140$ & $0.025$ \\
            \multirow{-2}{*}{Matrix gamma} & C2 & $0.942$ ($0.101$) & $0.891$ ($0.125$) & $-0.465$ & $0.352$ & $0.340$ && $-0.347$ & $0.333$ & $0.230$ \\
            \hline
        \end{tabular}
        }
    \end{center}
\end{table}
%========================================================%

%%%%%%%%%%%%%%%%%%%%%%%%%%%%%%%%%%%%%%%%%%%%%%%%%%%%%%%%%%

%%%%%%%%%%%%%%%%%%%%%%%%%%%%%%%%%%%%%%%%%%%%%%%%%%%%%%%%%%
% HCP Aging
%%%%%%%%%%%%%%%%%%%%%%%%%%%%%%%%%%%%%%%%%%%%%%%%%%%%%%%%%%
\section{Additional results of the HCP Aging study}
\label{appendix:sec:fmri}

%========================================================%
\subsection{Validity of model assumptions}
\label{appendix:subsec:fmri_assumption}

In this section, the validity of model assumptions imposed in Section~\ref{sub:asmp} is examined. In the HCP Aging study data, the number of observations of each subject is $u_{i}=u=478$ in the resting-state data and the number of observations of each subject is $v_{i}=v=335$ in the task-based data. The data dimensions are $p=q=75$. The sample size is $n=551$. Thus, Assumption A1 is satisfied. As discussed in Section~\ref{appendix:sub:thm_asmp} above, Assumption A2 is valid. Assumption A3 assumes common eigenstructures across covariance matrices. Here, we provide an empirical examination. First, the average sample covariance matrices, $\bar{\bS}_{y}$ and $\bar{\bS}_{x}$, are calculated and the eigenvectors are obtained. Second, eigenvectors of $\hat{\bSigma}_{i}$ and $\hat{\bDelta}_{i}$ of each subject are calculated. Correlations between individual eigenvectors and average eigenvectors are then calculated as a similarity metric. When the magnitude of the correlation is greater than $0.5$, it is considered as a high similarity allowing variability and bias in sample eigenvectors. In both $\bSigma_{i}$ and $\bDelta_{i}$, the first eigenvector is common across over $98\%$ of the subjects. The second to the fourth eigenvectors are common over $40\%$ of the subjects. This suggests that at least the partial common diagonalization assumption is satisfied in this dataset. Based on our simulation results in Section~\ref{sec:sim}, the proposed approach can identify the target components of interest. For Assumption A4 that the model is correctly specified, it is challenging to validate using data alone. The considered model is based on the domain knowledge and the study interest.
%========================================================%

%========================================================%
\subsection{Additional results}
\label{appendix:subsec:fmri_results}

Figure~\ref{appendix:fig:hcp} presents the sparsified loading profile of the identified components (C1, C2, and C3). Local smoothness and consistency is imposed within each brain functional module by using the fused lasso penalty~\citep{tibshirani2005sparsity}.
Figure~\ref{appendix:fig:hcp_score} presents the scatter plot of the modeled data after projection and the fitted regression line. From the figures, the linear association between the task-based functional connectivity within the network and the resting-state functional connectivity within the network is observed.

\begin{figure}
    \begin{center}
        \subfloat[C1-tb loading ($\bgamma$)]{\includegraphics[width=0.5\textwidth]{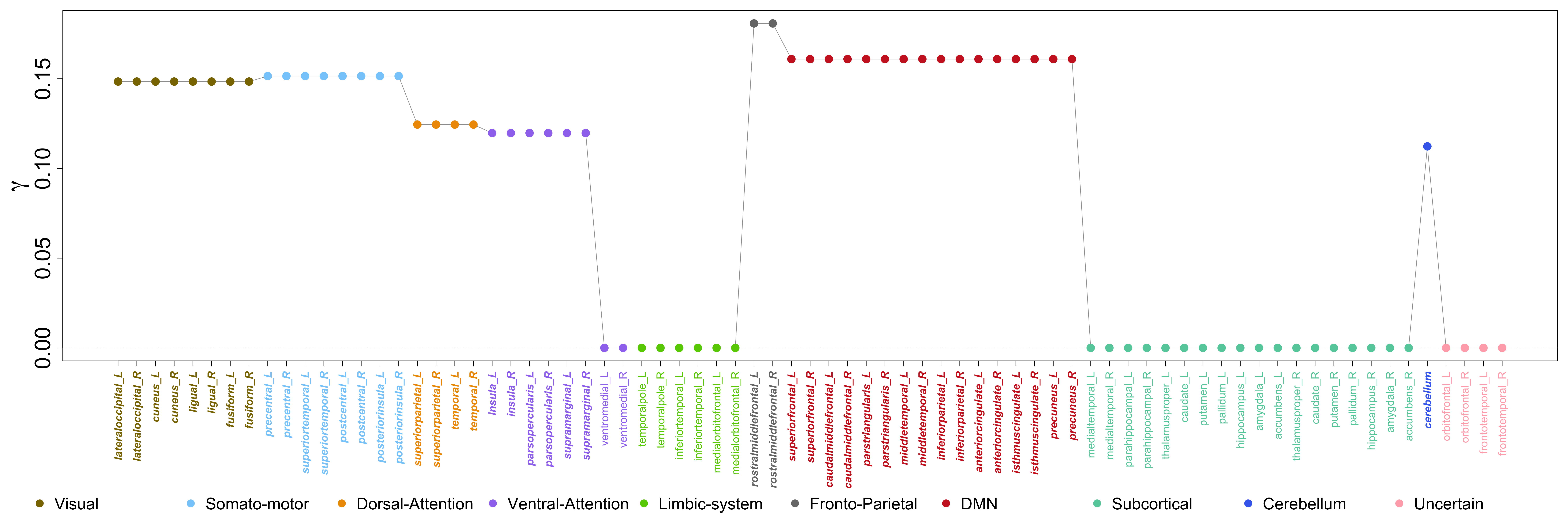}}
        \subfloat[C1-rs loading ($\btheta$)]{\includegraphics[width=0.5\textwidth]{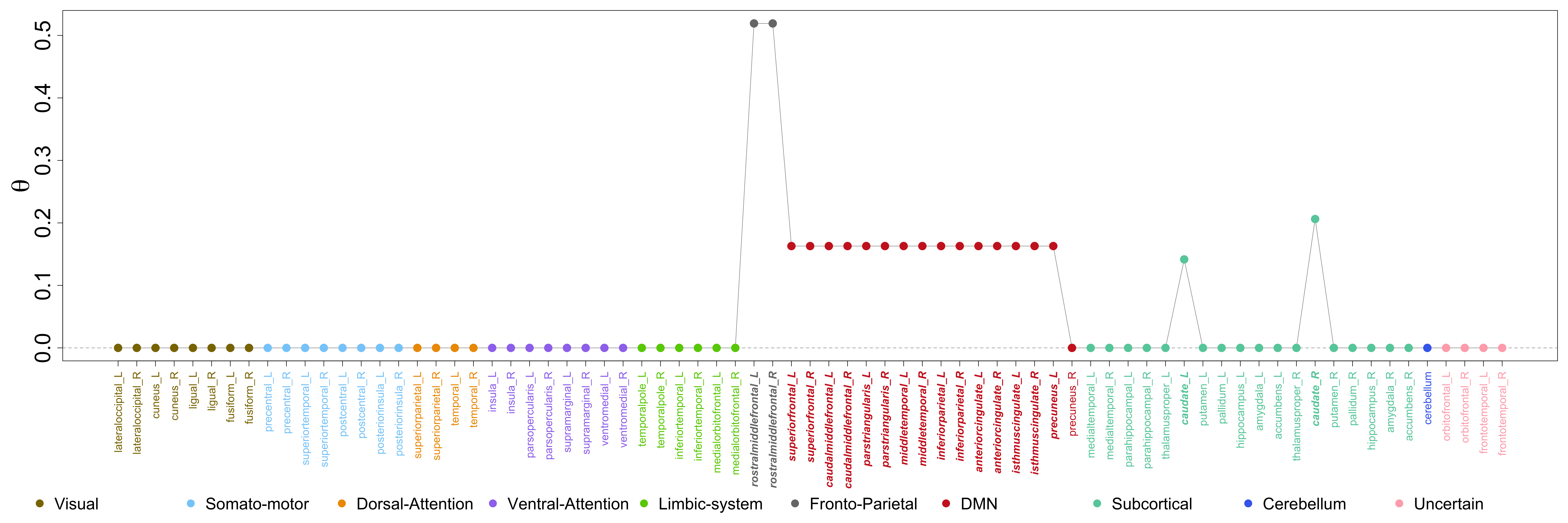}}

        \subfloat[C2-tb loading ($\bgamma$)]{\includegraphics[width=0.5\textwidth]{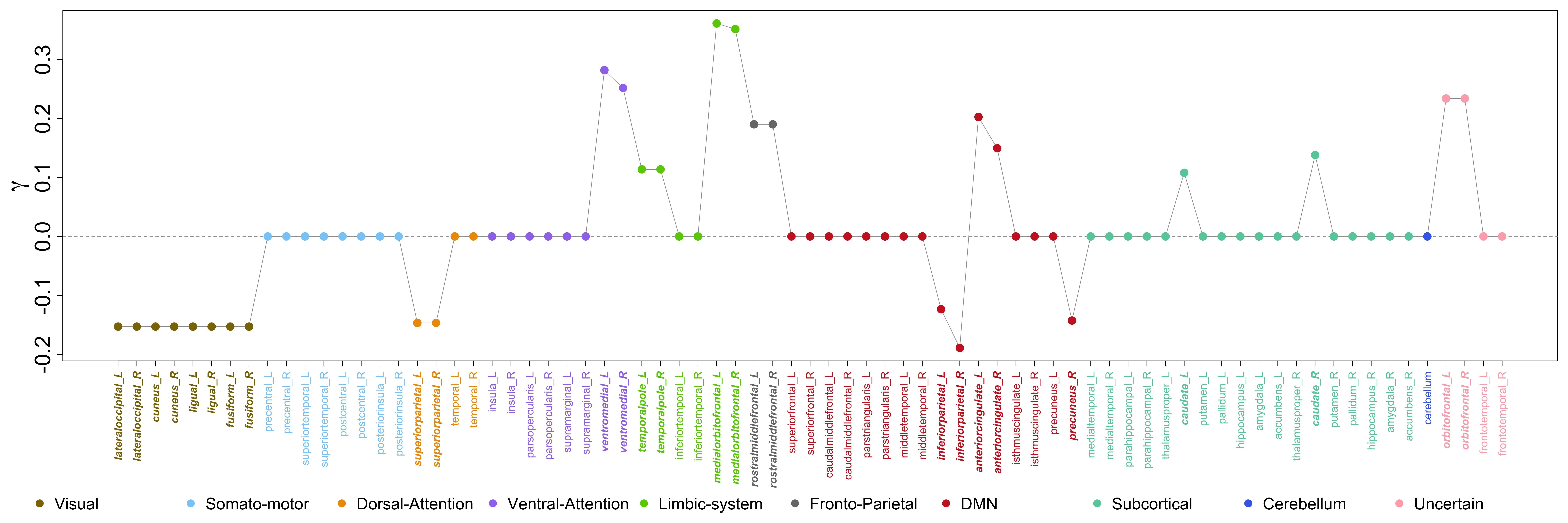}}
        \subfloat[C2-rs loading ($\btheta$)]{\includegraphics[width=0.5\textwidth]{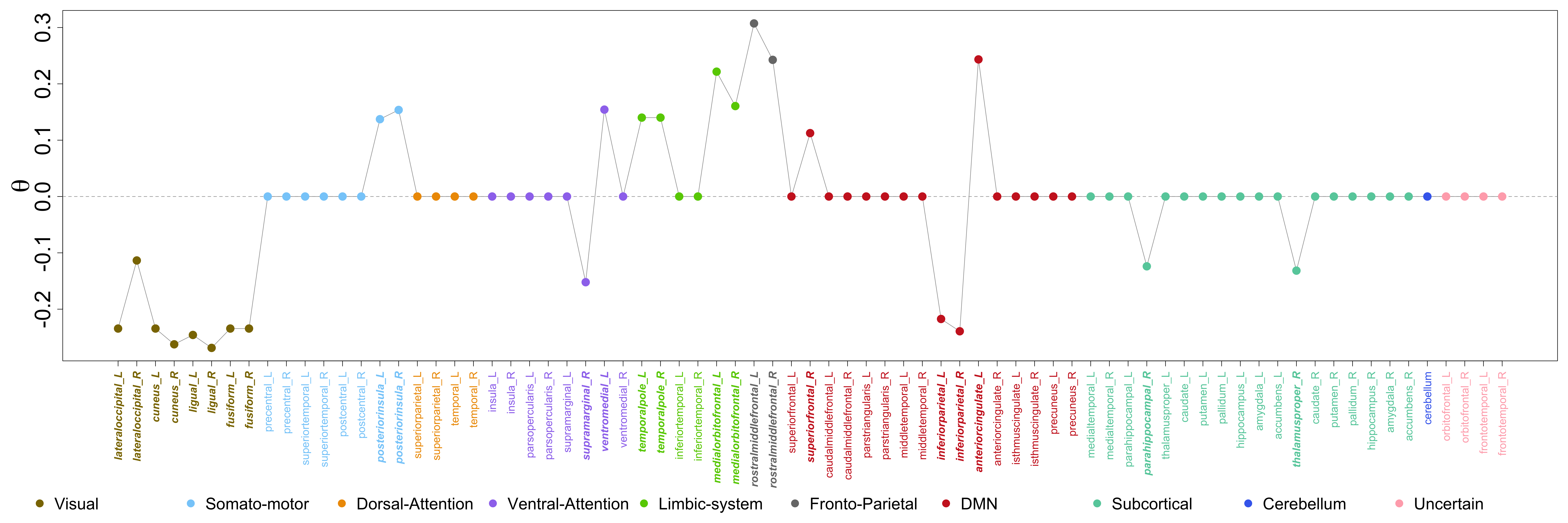}}

        \subfloat[C3-tb loading ($\bgamma$)]{\includegraphics[width=0.5\textwidth]{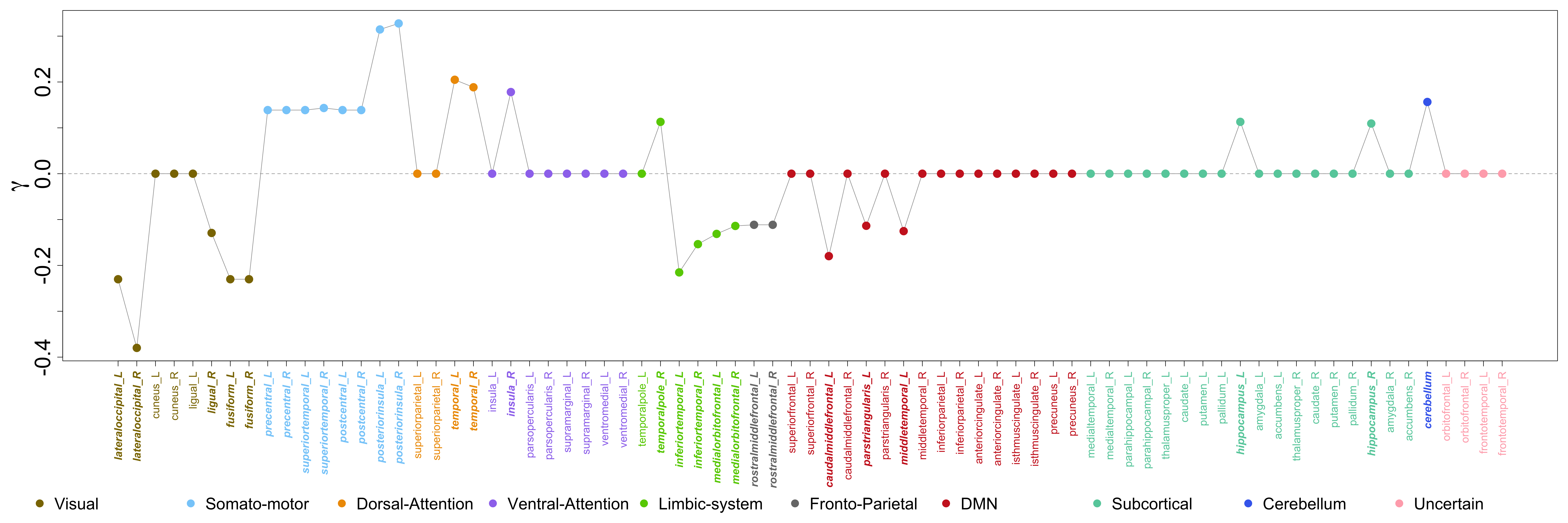}}
        \subfloat[C3-rs loading ($\btheta$)]{\includegraphics[width=0.5\textwidth]{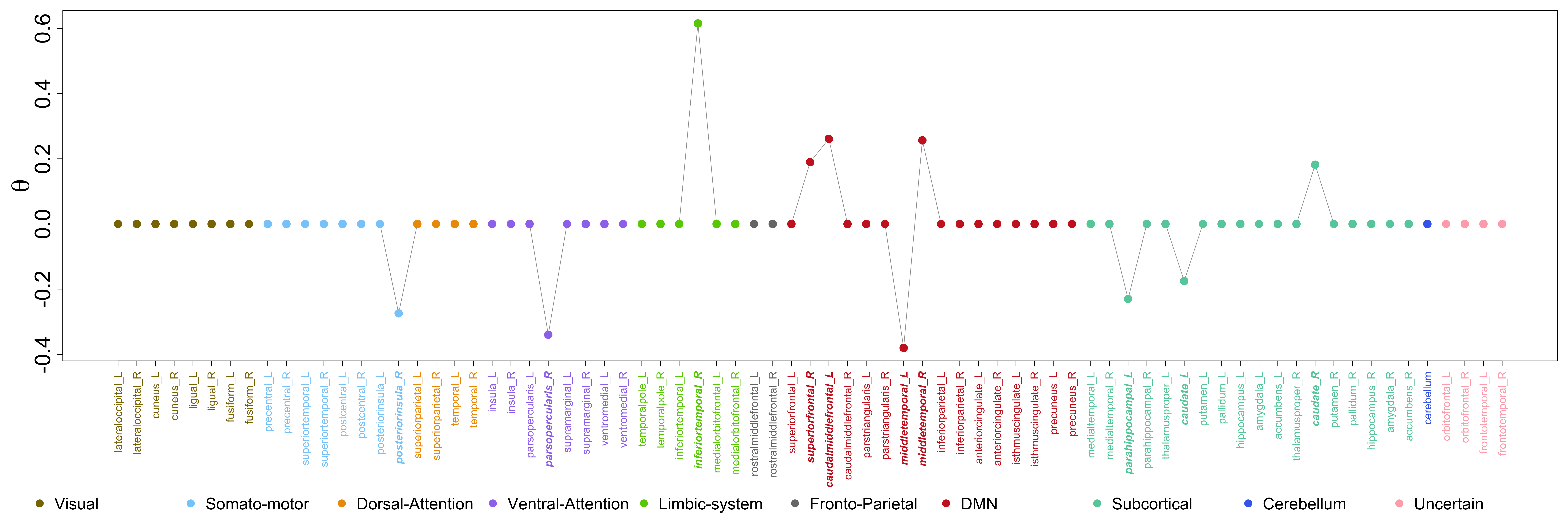}}
    \end{center}
    \caption{\label{appendix:fig:hcp}Sparsified loading profile of the identified components (C1, C2, and C3), using the proposed CoCReg approach in the HCP Aging study. tb: results of task-based fMRI; rs: results of resting-state fMRI.}
\end{figure}
\begin{figure}
    \begin{center}
        \subfloat[C1]{\includegraphics[width=0.33\textwidth]{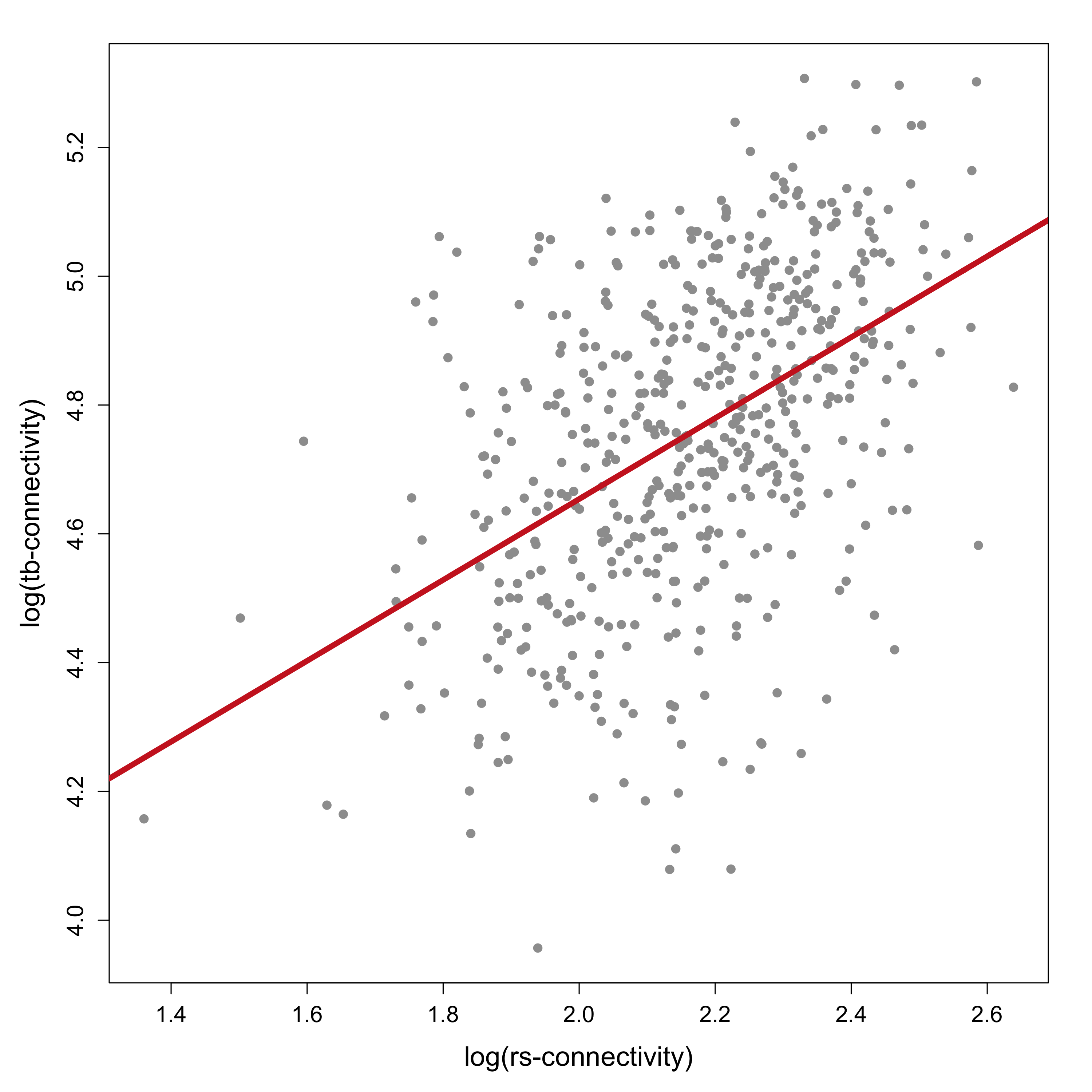}}
        \subfloat[C2]{\includegraphics[width=0.33\textwidth]{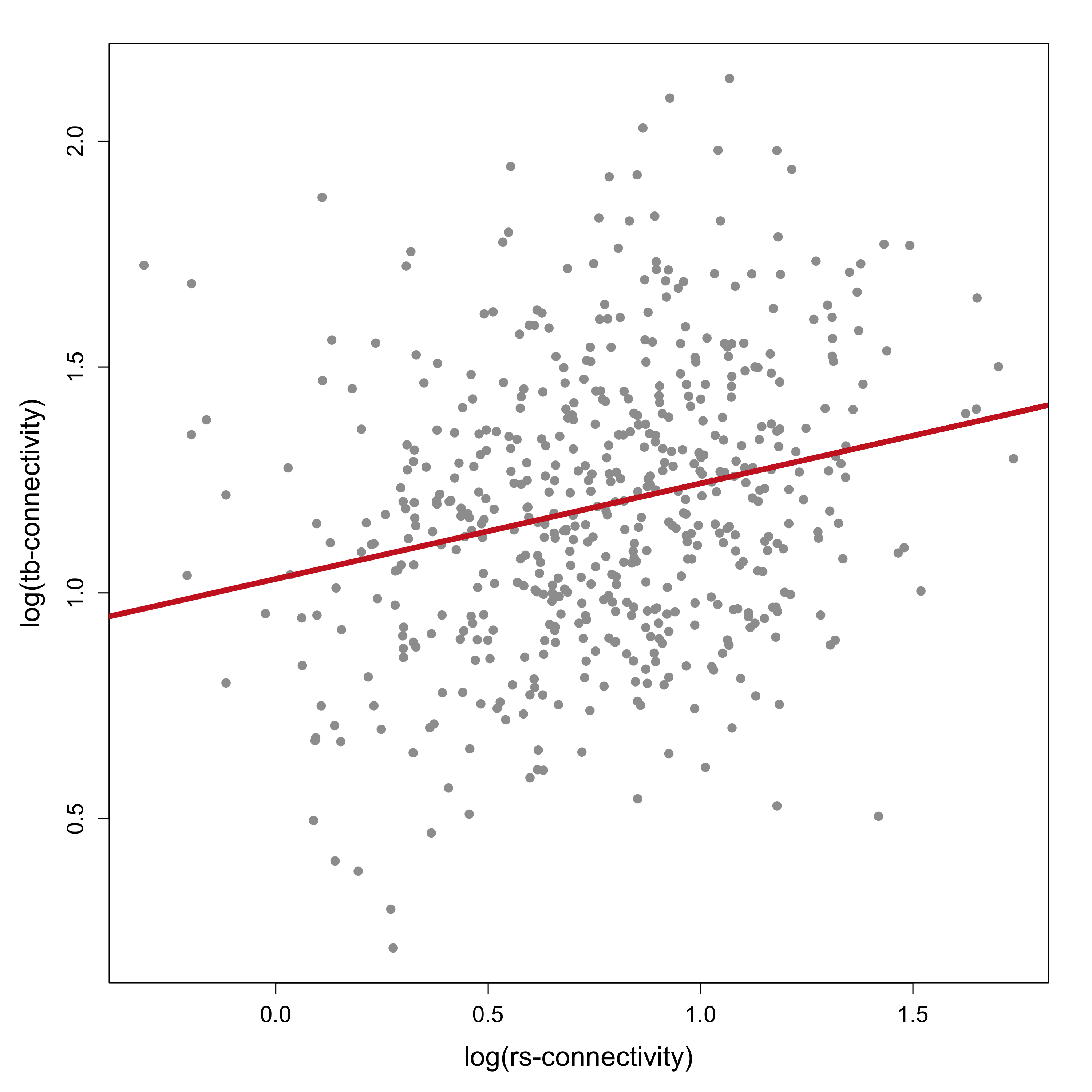}}
        \subfloat[C3]{\includegraphics[width=0.33\textwidth]{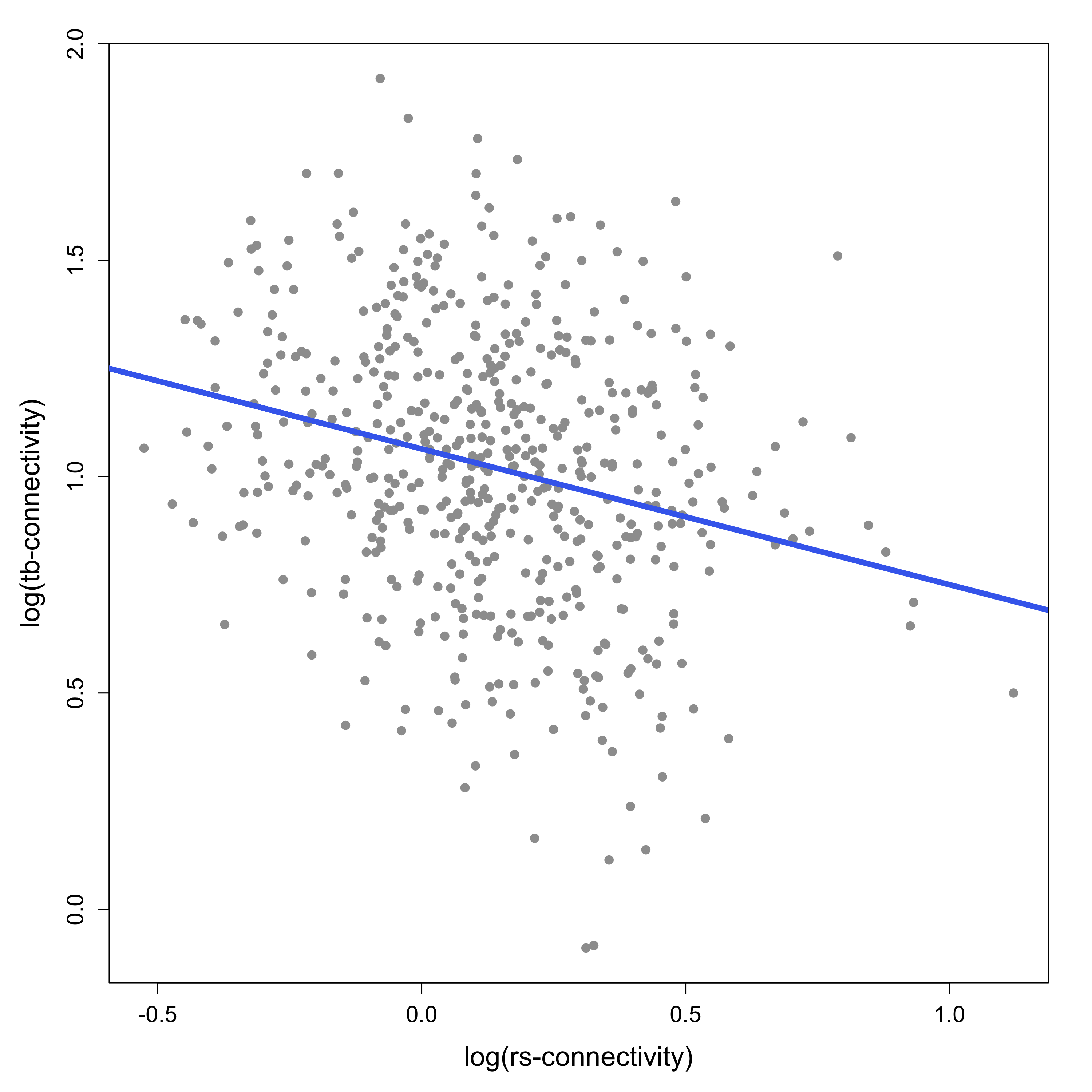}}
    \end{center}
    \caption{\label{appendix:fig:hcp_score}Scatter plot of $\log(\hat{\bgamma}^\top\hat{\bSigma}_{i}\bgamma)$ versus $\log(\hat{\btheta}^\top\hat{\bDelta}_{i}\btheta)$, as well as the fitted regression line, of the identified components (C1, C2, and C3), using the proposed CoCReg approach in the HPC Aging study. The value of $\log(\hat{\bgamma}^\top\hat{\bSigma}_{i}\bgamma)$ is adjusted for age and sex. tb: results of task-based fMRI; rs: results of resting-state fMRI.}
\end{figure}
%========================================================%
%%%%%%%%%%%%%%%%%%%%%%%%%%%%%%%%%%%%%%%%%%%%%%%%%%%%%%%%%%

%%%%%%%%%%%%%%%%%%%%%%%%%%%%%%%%%%%%%%%%%%%%%%%%%%%%%%%%%%
%========================================================%
% Reference
%========================================================%
% \clearpage

\bibliographystyle{apalike}
\bibliography{Bibliography}
%========================================================%
%%%%%%%%%%%%%%%%%%%%%%%%%%%%%%%%%%%%%%%%%%%%%%%%%%%%%%%%%%

%%%%%%%%%%%%%%%%%%%%%%%%%%%%%%%%%%%%%%%%%%%%%%%%%%%%%%%%%%
%%%%%%%%%%%%%%%%%%%%%%%%%%%%%%%%%%%%%%%%%%%%%%%%%%%%%%%%%%
\end{document}